\documentclass[12pt]{article}
\usepackage[left=2.5cm,right=2cm,
    top=2.5cm,bottom=2cm,bindingoffset=0cm]{geometry}
\usepackage{amsthm}
\usepackage{graphicx,textcomp}
\usepackage{amsmath,bm,amsfonts,mathrsfs,amssymb}
\usepackage{hyperref}
\newtheorem{theorem}{Theorem}
\newtheorem{lemma}[theorem]{Lemma}
\newtheorem{prop}[theorem]{Proposition}

\usepackage{appendix}
\usepackage[normalem]{ulem}
\usepackage{graphicx,color}
\usepackage{enumitem}

\begin{document}
\title{Determination of period matrix of double of surface with boundary via its DN map}
\author{Dmitrii Korikov\thanks{St.Petersburg Department of Steklov Mathematical Institute, St. Petersburg, Russia,
        \newline
        e-mail: dmitrii.v.korikov@gmail.com,
        \newline
        ORCID: \url{https://orcid.org/0000-0002-3212-5874}}.}
\maketitle
\begin{abstract}
As is well-known, a conformal class of a surface $M$ with boundary $\Gamma$ is determined by its DN map $\Lambda$. In the paper, the algorithm for determination of the $b$-period matrix $\mathbb{B}$ of the (Schottky) double of surface with boundary via $\Lambda$ is presented. Due to the Torelli theorem, $\mathbb{B}$ contains all information on the conformal class of $M$ except the proper way of attaching $\Gamma$ to it.
\end{abstract}
\smallskip

\noindent{\bf Keywords:}\,\,\,electric impedance tomography of
surfaces, $b$-period matrix, Di\-rich\-let-to-Neumann map,
stability of determination, moduli space.

\smallskip


\section{Introduction}
\paragraph{EIT problem.} Let $(M,g)$ be a surface (smooth oriented two-dimensional compact manifold) with (smooth) boundary $\Gamma$ diffeomorphic to a circle and smooth metric $g$. Let $\Delta$ be the Laplace-Beltrami operator on $(M,g)$; denote by $u^f$ the harmonic extension of the function $f\in H^{1/2}(\Gamma)$ into $M$. Let $\nu$ be the exterior normal vector on $\Gamma$. The continuous operator $\Lambda: \ H^{1/2}(\Gamma)\mapsto H^{-1/2}(\Gamma)$ defined by $\Lambda f:=\partial_\nu u^f$  is called the {\it Diriclet-to-Neumann} (DN) map. The two-dimensional {\it Electric Impedance Tomography} (EIT) problem consists in the determination of an unknown surface via its DN map.

It is well-known that the DN map $\Lambda$ determines only the conformal class of $(M,g)$ and the restriction of the metric to the boundary $\Gamma$. Namely, let $(M,g)$ and $(M',g')$ be two surfaces with the common boundary $\Gamma=\partial M=\partial M'$. We write $[(M,g)]=[(M',g')]$ if there is a conformal diffeomorphism between $(M,g)$ and $(M',g')$ which does not move the points of $\Gamma$. Then the theorem of Lassas and Uhlmann \cite{LU} states that $\Lambda=\Lambda'$ if and only if $[(M,g)]=[(M',g')]$ and $g$ and $g'$ induce the same length element on $\Gamma$. So, it is natural to understand the conformal class $[(M,g)]=:\mathscr{R}(\Lambda)$ as a solution to the EIT problem.

In \cite{BK_stab,K_stab}, the following natural result on the stability of solutions to the EIT problem is established. Let $\beta: \ M\mapsto M'$ be an orientation-preserving diffeomorphism and $x\in M$, then its differential $d\beta$ maps the unit circle (in the metric $g$) in $T_x M$ to some ellipse in $T_{\beta(x)}M'$ with major and minor semi-axes $r_>(x)$ and $r_<(x)$ (in the metric $g'$), respectively. The ratio $K_\beta(x)=r_>(x)/r_<(x)$ is called the {\it dilatation} of the map $\beta$ at $x$ while its maximum $K_\beta=\max_{x\in M}K_\beta(x)$ on $M$ is called the dilatation of $\beta$. Since $K_\beta=1$ if and only if $\beta$ is conformal, the quantity ${\rm log}K_\beta$ is a deviation of the map $\beta$ from being conformal. The {\it Teichmuller distance} between conformal classes $\tau=[(M,g)]$ and $\tau'=[(M',g')]$ is defined by 
\begin{equation}
\label{Teichmuller distance}
d_T(\tau,\tau'):=\frac{1}{2}\inf_\beta{\rm log}K_\beta,
\end{equation}
where the infimum is taken over all orientation-preserving diffeomorphisms from $M$ onto $M'$ which do not move the points of the common boundary $\Gamma$. Then $d_T$ is well-defined on conformal classes (i.e., it does not depend on the choice of the surfaces $(M,g)$ and $(M',g')$ representing $\tau$ and $\tau'$) and it is a {\it metric} on the space $\mathfrak{M}_{\mathfrak{g},\Gamma}$ of conformal classes $[(M,g)]$ of surfaces $(M,g)$ with given genus $\mathfrak{g}$ and boundary $\Gamma$. Note that, in the case $\Gamma=\varnothing$, the above definitions coincide with the definition of the classical {\it Teichm\"uller moduli space} $\mathcal{M}_\mathfrak{g}\equiv\mathfrak{M}_{\mathfrak{g},\varnothing}$ (see \cite{Alfors,Gardiner,Nag}). At the same time, the space $\mathscr{D}_{\mathfrak{g},\Gamma}$ of the DN maps of surfaces with genus $\mathfrak{g}$ and boundary $\Gamma$ is endowed with the metric given by the operator norm of the difference $d_{O}(\Lambda,\Lambda'):=\|\Lambda'-\Lambda\|_{H^{1/2}(\Gamma)\mapsto H^{-1/2}(\Gamma)}$. Then the stability result of \cite{BK_stab} states that the solving map $\mathscr{R}: \ \mathscr{D}_{\mathfrak{g},\Gamma}\to\mathscr{M}_{\mathfrak{g},\Gamma}$ is {\it continuous}. In other words, the closeness of $\Lambda'$ to $\Lambda$ implies the existence of a near-conformal diffeomorphism between $(M,g)$ and $(M',g')$ which does not move the points of $\Gamma$. This result is generalized in \cite{K_stab} for the non-orientable case and the case in which DN map is given only on a segment of the boundary. In addition, in \cite{K_stab}, it is proved that the map $\mathscr{R}: \ \mathscr{D}_{\mathfrak{g},\Gamma}\to\mathscr{M}_{\mathfrak{g},\Gamma}$ and its inverse are point Lipschitz continuous, i.e., the following local stability estimate holds
\begin{equation}
\label{stability estimate}
c(\Lambda)d_{O}(\Lambda,\Lambda')\le d_T(\mathscr{R}(\Lambda),\mathscr{R}(\Lambda'))\le C(\Lambda)d_{O}(\Lambda,\Lambda') \qquad (d_{O}(\Lambda,\Lambda')\le R(\Lambda))
\end{equation}
(here the positive constants $c(\Lambda),C(\Lambda),R(\Lambda)$ depend only on $\Lambda$). Note that, in the above stability result, both $\Lambda$ and $\Lambda'$ are assumed to be DN maps (and the corresponding surfaces are homeomorphic); the case of noisy boundary data was not discussed here.

\paragraph{Main result.} One may wonder if there is a more explicit connection between $\Lambda$ and $\mathscr{R}(\Lambda)$ that extends formula (\ref{stability estimate}) (e.g. differentiability/explicit formulas for conformal invariants of $(M,g)$, etc.)? However, the moduli spaces $\mathfrak{M}_{\mathfrak{g},\Gamma}$ of surfaces with fixed boundary $\Gamma$ are not finite-dimensional and thus are inconvenient for these purposes. This is due to the presence of infinitely many degrees of freedom related to different ways of attaching a surface to the curve $\Gamma$ (in other words, infinitely many reparametrizations of DN maps
$$\Lambda\mapsto\Lambda_\phi, \qquad \Lambda_\phi f:=(\Lambda(f\circ\phi^{-1}))\circ\phi,$$
where $\phi$ is an arbitrary diffeomorphism of $\Gamma$). One can get rid of these ``extra'' degrees of freedom by considering the (Schottky) {\it double} $\mathbb{M}$ of the surface $M$ which is the Riemann surface without boundary obtained by gluing two copies $M\times\{\pm\}$ of $(M,g)$ along the boundaries (i.e., by the identification $(x\times +)\sim (x\times -)$ of points $x\times +$ and $x\times -$, where $x\in\Gamma$). The double $\mathbb{M}$ is endowed with the anti-holomorphic involution $\tau: \ (x\times\pm)/\sim \ \mapsto (x\times\mp)/\sim$. One can identify $M$ with one of the submanifolds $(M\times\{\pm\})/\sim$ obtained by cutting $\mathbb{M}$ along the curve $\{x\in\mathbb{M} \ | \ \tau(x)=x\}$. Denote the conformal class of $\mathbb{M}$ by $\hat{\mathscr{R}}(\Lambda)$, where $\Lambda$ is the DN map of $M$. Then $\Lambda'=\Lambda$ (or even $\Lambda'=\Lambda_\phi$) implies $\hat{\mathscr{R}}(\Lambda)=\hat{\mathscr{R}}(\Lambda')$ and, due to the definition of the Teichm\"uller distance, inequality (\ref{stability estimate}) implies
\begin{equation}
\label{stability cover}
d_T(\hat{\mathscr{R}}(\Lambda),\hat{\mathscr{R}}(\Lambda'))\le C(\Lambda)d_{O}(\Lambda,\Lambda') \qquad (d_{O}(\Lambda,\Lambda')\le R(\Lambda)).
\end{equation}

The moduli space $\mathcal{M}_m$ of the surfaces of genus $m>1$ without boundaries is a complex $(3m-3)$-dimensional orbifold while the conformal classes of doubles of genus $\mathfrak{g}$ surfaces with boundaries diffeomorphic to a circle constitute the stratum $\mathcal{M}^{\circ}_\mathfrak{g}$ of real dimension $6\mathfrak{g}-3$ in $\mathcal{M}_{2\mathfrak{g}}$ \cite{FK}. Thereby, the original EIT problem $\Lambda\mapsto\mathscr{R}(\Lambda)$ is replaced by the finite-dimensional {\it reduced EIT problem} $\Lambda\mapsto\hat{\mathscr{R}}(\Lambda)$ which consists in determination of the appropriate coordinates of the double of $[(M,g)]$ in the moduli space via its DN map $\Lambda$. 

Most of the known (say, Fenchel–Nielsen's) local coordinates on the moduli space are highly dependent on the methods of their construction and are therefore inconvenient for the reduced EIT problem. The exception is the coordinates provided by the entries of {\it $b$-period matrices} of surfaces. Recall that a {\it Torelli marked surface} is a Riemann surface $X$ without boundary equipped with a choice of canonical homology basis $[l_\cdot]=\{a_1,\dots,a_{m},b_1,\dots,b_{m}\}$ (``marking'') on it. We say that two Torelli marked surfaces $(X,[l_\cdot])$ and $(X',[l'_\cdot])$ are equivalent if there is a biholomorphism $\beta$ between them which preserves the marking (i.e., $\beta\circ a_k=a'_k$, $\beta\circ b_k=b'_k$). The space $\mathcal{T}_{m}$ of equivalence classes of Torelli marked surfaces of genus $m$ (endowed with metric (\ref{Teichmuller distance}), where the infimum is taken over all marking-preserving diffeomorphisms) is the infinite-sheeted covering space of the moduli space $\mathcal{M}_m$ called the {\it Torelli space}. Let $\omega_1,\dots,\omega_m$ be the basis of holomorphic differentials on $X$ dual to the homology basis (i.e., thier periods obeys $T(\omega_i,a_j):=\int_{a_j}\omega_i=\delta_{ij}$). Then the $m\times m$-matrix $\mathbb{B}$ with the entries 
$$\mathbb{B}_{ij}=T(\omega_i,b_j):=\int_{b_j}\omega_i$$ 
is called the $b$-{\it period matrix} of the Torelli marked surface $(X,[l_\cdot])$. It is clear that $\mathbb{B}$ is a conformal invariant, i.e., it depends only on the class $[(X,[l_\cdot])]$ of $(X,[l_\cdot])$ in $\mathcal{T}_{m}$. Due to the {\it Torelli theorem} (\cite{Torelli}, see also \cite{GH,Milne,Narasimhan}), the $b$-period matrix $\mathbb{B}$ determines $[(X,[l_\cdot])]$, i.e., the map $[(X,[l_\cdot])]\mapsto\mathbb{B}$ is an injection. So, the entries of the $b$-period matrix are indeed the local coordinates on $\mathcal{M}_{m}$. Note that although the $b$-period matrix of $X$ is not uniquely determined by its conformal class $[X]\in\mathcal{M}_{m}$ due to the infinitely many choices of marking on $X$, any two $b$-period matrices of $X$ are related to each other via well-known transformations corresponding to the change of the canonical homology basis. In addition, the $b$-period matrices of surfaces of genus $m$ belong to the Siegel upper half-space $\mathcal{H}_m$ (the space of symmetric matrices with positive-definite imaginary parts) of the dimension $m(m+1)/2$ while the dimension of $\mathcal{M}_m$ is $3m-3$. Thus, the entries of the $b$-period matrix are not independent for higher genera $m>3$. In particular, the solutions to the reduced EIT problem (elements of $\mathcal{M}^{\circ}_{\mathfrak{g}}$) are described by $6\mathfrak{g}-3$ real parameters while their $b$-period matrices (considered as elements of $\mathcal{H}_{2\mathfrak{g}}$) provide $2\mathfrak{g}(2\mathfrak{g}+1)$ real parameters.

The main result of the paper is the {\it algorithm} for deriving the $b$-period matrix of the double $\mathbb{M}$ of the surface $(M,g)$ via its DN map $\Lambda$. It is presented at Steps 1-4, Section \ref{Algorithm sec}. The first (more or less standard) step is determining the boundary data associated with the Abelian differentials on the double $\mathbb{M}$ of $(M,g)$. As a result, we obtain the isomorphic copy (endowed with additional structures like inner product, e.t.c.) of the space $H^0(\mathbb{M};K)$ of Abelian differentails on $\mathbb{M}$ (Proposition \ref{determination of normal harmonic fields} and Lemma \ref{Connection between normal fields and symmetric differentials}). The second step, which is the key in our procedure, is the determination of the boundary data associated with the Abelian differentials whose periods have {\it integer} imaginary parts. Here the key trick (Proposition \ref{integer periods condition prop}) is reducing such a determination to solving the non-linear equations
\begin{equation}
\label{Main equation on the coefficients}
\partial_\gamma(H-i)\big[p_1^{\alpha_1}\dots p_{\mathfrak{g}}^{\alpha_\mathfrak{g}}q_1^{\beta_1}\dots q_{\mathfrak{g}}^{\beta_\mathfrak{g}}\big]=0.
\end{equation}
on the unknown real parameters $\alpha_1,\dots,\alpha_{\mathfrak{g}},\beta_1,\dots,\beta_{\mathfrak{g}}$. Here $\partial_\gamma$ is the differentiation along $\Gamma$ and $H:=\Lambda^{-1}\partial_\gamma$ is the {\it Hilbert transform} of the surface $(M,g)$; the functions $p_1,\dots,p_{\mathfrak{g}},q_1,\dots,q_{\mathfrak{g}}$ are determined by the eigenfunctions and eigenvalues of $H$ (via formula (\ref{sepatrate factors}) below). On the third step, we apply Proposition \ref{improved criterion for canonicity prop}) to construct the isomorphic copy of the basis in $H^0(\mathbb{M};K)$ dual to some canonical homology basis $[l_\cdot]$ on $\mathbb{M}$. On the last step, we calculate the $b$-period matrix $\mathbb{B}$ in this homology basis. By applying trivial transformations, one can obtain from $\mathbb{B}$ all other $b$-period matrices of $\mathbb{M}$. It worth noting that, although the homology basis $[l_\cdot]$ is unknown, it obeys an additional symmetry property (see formula (\ref{period matrix symmetries}) below) with respect to the involution on the double $\mathbb{M}$.

\paragraph{Comments.} 1) In its traditional understanding, the two-dimensional EIT consists in the construction (or the visualization) of some conformal copy of the surface with given DN map. There are several approaches to perform this. The method of \cite{LU} is based on the simultaneous analytic continuation of harmonic functions from the boundary in the coordinates provided by each other. In the algebraic approach of \cite{B}, the conformal copy of a surface is constructed as the spectrum (the set of multiplucative linear functionals) of the algebra of holomorphic functions on the surface; the latter being determined up to isomorphism by the DN map. As follows from the descriptions, both approaches are highly abstract and thus unsuitable for surface visualization. The method of \cite{HM,Michel} allows to construct the conformal copy as a part of an algebraic curve immersed in $\mathbb{CP}^2$ and thus is most appropriate for the visualization; however, this algorithm seems to be highly unstable under small perturbations of the DN map. The method of \cite{BK_stab_imm,BK_stab,K_stab} makes use of holomorphic embeddings into high-dimensional spaces $\mathbb{C}^n$ instead, which leads to the proof of the (Teichm\"uller) stability of solutions to the EIT problem. However, the applicability of the methods of these papers as an algorithm for construction a copy (including the stability of the solutions in the presence of noisy boundary data) has not been studied.

2) In contrast to the above approaches, we deal with the calculation of numerical parameters that encode the most informaton on the unknown surface $(M,g)$ including the conformal structure on it. Indeed, in view of the Torelli theorem \cite{Torelli}, the $b$-period matrix $\mathbb{B}$ determines (up to biholomorphism) the double $\mathbb{M}$ of $(M,g)$. In the generic case, $\mathbb{M}$ admits the unique antiholomorphic involution (the surfaces with several of them constitute the lower-dimensional stratum). Even in exceptional cases, the additional symmetry of $\mathbb{B}$ (provided by (\ref{period matrix symmetries})) allows one to choose the proper involution $\tau$ on $\mathbb{M}$. Now the cutting $\mathbb{M}$ along the set of fixed points of $\tau$ provides two conformal copies $M'$, $M''$ of $(M,g)$. So, the only information that is lost is the proper way of identifying of the points of the curve $\Gamma$ and the points of the boundary of $M'$ and the proper choice of the metric on the boundary with which $\Lambda$ becomes a DN map of $M'$. Although this information could in principle be obtained by including additional steps in the algorithm, this question is not covered in the present paper.

3) As showed in \cite{B}, the genus of the surface is determined by its DN map. Namely, if $H=\partial_\gamma\Lambda^{-1}$ is the Hilbert transform of the surface $M$, then its genus ${\rm gen}(M)$ is just the total multiplicity of eigenvalues of $H$ contained in $\mathbb{C}_+\backslash\{i\}$. It worth noting that the surface genus is not stable under small perturbations of its DN map \cite{K_stab_top}. Namely, by cutting small disks from $(M,g)$ and attaching a finite number $k$ of small handles, one provides the higher genus surface whose DN map is arbitrarily close to $\Lambda$. In this case, the $k$ ``extra'' eigenvalues of $H$ in $\mathbb{C}_+\backslash\{i\}$ are close to $i$. Note that one cannot lower the surface genus without significant change of its DN map.

4) As shown in \cite{BSh}, the real additive cohomology structure of the manifold with boundary is determined by its DN map defined on exterior differential forms. This result is improved in \cite{Shonkwiler} where it is proved that the information on the multiplicative structure (the cap product) of cohomologies can be also recovered from the DN map. Also, in \cite{Shonkwiler}, a simple connection between the eigenvalues of the Hilbert transform and Poincar\'e duality angles of the manifold is established. The methods of \cite{BSh,Shonkwiler} have much in common with Step 1 of the present algorithm.

5) The continuous dependence of the $b$-period matrix $\mathbb{B}$ of $\mathbb{M}=\hat{\mathscr{R}}(\Lambda)$ on the DN map $\Lambda\in\mathcal{D}_{\mathfrak{g},\Gamma}$ (provided the appropriate choice of marking on $\mathbb{M}$) trivially follows from estimate (\ref{stability cover}). The {\it stability of the algorithm} for determining $\mathbb{B}$ in the presence of small noise in the boundary data is discussed in the end of Section \ref{Algorithm sec}. There we also prove the following convergence-type stability result.
\begin{prop}
\label{convergence prop}
Let $\Lambda$ be a fixed DN map of some surface $(M,g)$ of genus $\mathfrak{g}$ with {\rm(}known{\rm)} boundary $\Gamma$. Then there are sufficiently small numbers $\varepsilon_0=\varepsilon_0(\Lambda)>0$ and $c_0=c_0(\Lambda)>0$ such that the implementation of the algorithm Steps 1-4 to any approximation $\Lambda'$ of $\Lambda$ obeying 
$$\|\Lambda'-\Lambda\|_{H^{1}(\Gamma)\to L_2(\Gamma)}=\varepsilon<\varepsilon_0$$
provides the matrix $\mathbb{B}'$ obeying
$$\|\mathbb{B}'-\mathbb{B}\|_{M^{2\mathfrak{g}\times 2\mathfrak{g}}}\le c_0\varepsilon,$$ 
where $\mathbb{B}$ is some $b$-period matrix of the double $\mathbb{M}$ of $(M,g)$. {\rm(}Note that the implementation of Steps 1-4 requires the a priori knowledge of the noise bound $\varepsilon$.{\rm)}
\end{prop}

\section{Preliminaries} 
\paragraph{Complex structure.} As is well-known, the orientation and the conformal class $[g]$ of metrics on $M$ determine the unique complex structure (biholomorphic sub-atlas of the smooth oriented atlas on $M$) on it, such that, in any holomorphic coordinate $z$ on $M$, the metric $g$ is of the form $g(z)=\rho(z)|dz|^2$, where $\rho(z)>0$ (equivalently, $(\star+i\,{\rm Id}) dz=0$, where $\star$ is the Hodge operator on $(M,g)$). Given this complex structure, a function $w$ on $M$ is holomorphic (resp., anti-holomorphic) if and only if the Cauchy-Riemann condition $d\Im w=\star d\Re w$ (resp., $d\Im w=-\star d\Re w$) holds. The space of functions holomorhic on ${\rm int}M$ and smooth up to the boundary $\Gamma$ is denoted by $\mathscr{A}(M)$.

The operator 
$$\Phi: \ A\mapsto (\star A^\flat)^\sharp$$ 
(here $\flat: \ TM\mapsto T^*M$ and $\sharp:=\flat^{-1}$ are the musical isomorphisms defined by $A^\flat:=g(A,\cdot)$) acts as the counterclockwise rotation on the right angle in each tangent space $T_x M$ ($x\in M$). Note that both $\star$ and $\Phi$ are independent of the choice of metric $g$ from the conformal class $[g]$. The Cauchy-Riemann condition can be rewritten as $\nabla\Im w=\Phi\Re\nabla w$ (in any metric from $[g]$). 

Choose the unit tangent vector $\gamma$ on $\Gamma$. In the subsequent, we agree that the orientations of $M$ and $\Gamma$ are related by 
\begin{equation}
\label{orientation of boundary}
\Phi\nu=\gamma.
\end{equation}

\paragraph{Harmonic fields.} Denote by $L_2(M;TM)$ the space of square integrable vector fields on $M$, endowed with the inner product $(A,B):=\int_M g(A,B)dS$. The harmonic fields constitute the (closed) subspace  
$$\mathcal{H}:=\{A\in  L_2(M;TM) \ | \ {\rm div}(\Phi A)={\rm div}A=0 \text{ in } M\}$$
in $L_2(M;TM)$. By definition, the rotation $\Phi$ is an isometric automorphisms of $L_2(M;TM)$ which preserve harmonicity and obeys $\Phi^{-1}=\Phi^*=-\Phi$. Also, each harmonic field $A$ on $M$ can be represented as $A=\nabla u$ (with harmonic $u$) in any simple-connected domain in $M$.

Introduce the subspace of potential fields $\mathcal{E}:=\{\nabla u\in\mathcal{H}\}$ and denote by $\mathcal{D}$ its orthogonal complement in $\mathcal{H}$. Let $\mathcal{N}=\Phi\mathcal{D}$. In view of the Stokes theorem, formula
$$(A,\nabla u)=\int\limits_M {\rm div}(uA)=\int\limits_\Gamma uA_\nu dl$$
holds for any  $A\in\mathcal{H}$ and $u\in C^\infty(\overline{M})$, where $A_\nu:=g(A,\nu)$. Thus, a harmonic field belongs to $\mathcal{D}$ ($\mathcal{N}$) if and only if it is tangent (normal) to $\Gamma$. In particular, any $A\in\mathcal{D}$ ($A\in\mathcal{N}$) is smooth up to the boundary due to the increasing smoothness theorems for solutions to elliptic boundary value problems. Note that 
\begin{equation}
\label{dim of defect space}
{\rm dim}\mathcal{D}={\rm dim}\mathcal{N}=2\mathfrak{g}.
\end{equation}
Denote $A_\gamma:=g(A,\gamma)$. Let $A\in\mathcal{D}$; then $\Phi A\in \mathcal{H}$ and the Stokes theorem yields $\int_{\Gamma} A_\gamma dl=-\int_M {\rm div}(\Phi A)dS=0$. Thus, each $A\in\mathcal{D}\cup\mathcal{N}$ can be represented as $A=\nabla u$ in a tubular neighborhood of $\Gamma$. In particular, the maps $\mathcal{D}\ni A\mapsto A_\gamma$, $\mathcal{N}\ni B\mapsto B_\nu$ are injections due to the uniqueness of solution to the Cauchy problem for the Laplace equation.

\paragraph{Hilbert transform.} Denote by $P$ the orthogonal projection on $\mathcal{E}$ in $L_2(M;TM)$ and introduce the reduced rotation $\hat{\Phi}:=P\Phi P$. Since $\Phi$ is anti-hermitian, so is $\hat{\Phi}$. In what follows, we also consider the complexification $\hat{\Phi}(A+iB)=\hat{\Phi}A+i\hat{\Phi}B$ of $\hat{\Phi}$ acting in $L^\mathbb{C}_2(M;TM)$. 

Let $u=u^f$ be a harmonic function in $M$ with trace $f$ on $\Lambda$. Then $\Phi\nabla u$ is a harmonic field. From the orthogonal decomposition $\mathcal{H}=\mathcal{E}\oplus\mathcal{D}$, we have 
\begin{equation}
\label{Helmholtz decomposition 0}
\Phi\nabla u^f=\nabla v^h+A,
\end{equation}
where $A\in\mathcal{D}$ and $v^h$ is some harmonic function with the trace $h$. Hence, $\hat{\Phi}\nabla u^f=\nabla v^h$ and $-\hat{\Phi}\nabla v^h=\nabla u^f+P\Phi A$, i.e., 
$$(\hat{\Phi}+iI)\nabla w=-P\Phi A, \text{ where } w=u^f+iv^h.$$ 
Note that $P\Phi A=0$ if and only if $A=0$. Indeed, if $P\Phi A=0$, then $(\Phi A,\nabla u)=(\Phi A,P\nabla u)=(P\Phi A,\nabla u)=0$ for any smooth $u$ on $M$. Since $\Phi A$ is harmonic, it means that 
$$0=\int_M{\rm div}(u\Phi A)dS=\int_\Gamma u(\Phi A)_\nu dl=-\int_\Gamma uA_\gamma dl,$$ 
whence $A_\gamma=0$ on $\Gamma$ and $A=0$ in $M$. In view of (\ref{Helmholtz decomposition 0}), we obtain the following criteria: $w$ is holomorphic (resp., antiholomorphic) in $M$ if and only if $\nabla w$ is an eigenvector of $\hat{\Phi}$ corresponding to the eigenvalue $-i$ (resp., $+i$). 

In view of (\ref{Helmholtz decomposition 0}), the equality $\hat{\Phi}\nabla u^f=0$ implies $\nabla v^h=\hat{\Phi}\nabla u^f=0$ and $\partial_\gamma f=-(\Phi\nabla u^f)_\nu=-\partial_\nu v^h+0=0$, whence $u^f={\rm const}$ and $\nabla u^f=0$. In addition, (\ref{Helmholtz decomposition 0}) implies that $(\hat{\Phi}-\Phi)\nabla u^f=0$ if and only if $A=0$, i.e., if and only if $\nabla u^f\in {\rm Ker}(\hat{\Phi}-i)\oplus{\rm Ker}(\hat{\Phi}+i)$. 

Let us show that 
\begin{equation}
\label{all defect vectors}
(\Phi-\hat{\Phi})\mathcal{E}=\mathcal{D}.
\end{equation}
Indeed, since the left-hand side is equal to $(I-P)\Phi\mathcal{E}$, it is contained in $\mathcal{D}$. Next, suppose that the field $A\in\mathcal{D}$ is orthogonal to $(I-P)\Phi\mathcal{E}$. Then $-(\Phi A,\nabla u)=(A,\Phi\nabla u)=(A,P\Phi\nabla u)=(PA,\Phi\nabla u)=0$ for any $u\in C^\infty(\overline{M})$ and $\Phi A\in \mathcal{D}$. The last equality means that $A_\gamma=-(\Phi A)_\nu=0$. Therefore, $A=0$ in $M$.

In view of the aforementioned, the eigenvalues of $\hat{\Phi}$ are $0$ (the corresponding eigenspace is $L^\mathbb{C}_2(M;TM)\ominus\mathcal{E}^\mathbb{C}$), $-i$ and $+i$ (the corresponding eigenspaces consist of gradients of holomorphic and anti-holomorphic functions on $M$, respectively) and the remaining eigenvalues have the total multiplicity ${\rm dim}(\Phi-\hat{\Phi})\mathcal{E}={\rm dim}\mathcal{D}=2\mathfrak{g}$ in view of (\ref{dim of defect space}). Since $\hat{\Phi}\nabla\overline{w}=\overline{\hat{\Phi}\nabla}w$ and $\hat{\Phi}$ is anti-hermitian, the remaining eigenvalues (counted with their multiplicities) can be represented as
\begin{equation}
\label{eigenvalues}
\lambda_{\pm k}=i\mu_{\pm k}, \qquad \mu_{\pm k}=-\mu_{\mp k}\in\mathbb{R} \qquad (k=1,\dots,\mathfrak{g}).
\end{equation}

Denote by $(\cdot,\cdot)_\Gamma$ the inner product in $L^{\mathbb{C}}_2(\Gamma;dl)$. Let $\langle f\rangle:=(f,1)_\Gamma/(1,1)_\Gamma$ denotes the mean value of $f$ on $(\Gamma,dl)$. In view of the Green formula
\begin{equation}
\label{Green formula}
(\nabla u^f,\nabla u^h)=(\Lambda f,h)_\Gamma=:(f,h)_\Lambda,
\end{equation}
the map $\mathfrak{E}: \ f\mapsto \nabla u^f$ is an isometry from the space $\partial_\gamma H^{3/2}(\Gamma;\mathbb{C})=\{f\in H^{1/2}(\Gamma;\mathbb{C}) \ | \ \langle f\rangle=0\}$ equipped with the inner product $(\cdot,\cdot)_\Lambda$ onto $\mathcal{E}^\mathbb{C}$. 

We define the {\it Hilbert transform} as the isomorphic copy 
$$H:=-\mathfrak{E}^{-1}\hat{\Phi}\mathfrak{E}$$ 
of the reduced rotation $\hat{\Phi}=P\Phi P$ (the minus sign is introduced to match the usual definition of the Hilbert transform on the circle). Then the Stokes theorem yields
\begin{align*}
(\Lambda Hf,h)_\Gamma=(Hf,h)_\Lambda=-(P\Phi P\nabla u^f,\nabla u^h)=&(\Phi \nabla u^f,\nabla u^h)=\\
=&\int_M {\rm div}(\overline{u^h}\Phi\nabla u^f)dS=(-\partial_\gamma f,h)_\Gamma
\end{align*}
for any $f,h\in \partial_\gamma H^{3/2}(\Gamma;\mathbb{C})$. Hence,
\begin{equation}
\label{Hilbert transfrom}
H=\Lambda^{-1}\partial_\gamma, \qquad H^{-1}=\partial^{-1}_\gamma\Lambda.
\end{equation}
Here $\partial_\gamma^{-1}$ is the integration with respect to the length element along $\Gamma$ in the direction $\gamma$. Since the images $\Lambda H^1(\Gamma;\mathbb{C})$ and $\partial_\gamma H^1(\Gamma;\mathbb{C})$ are orthogonal to constants in $L_2(\Gamma)$, both operators (\ref{Hilbert transfrom}) are well-defined on $H^1(\Gamma;\mathbb{C})$. In addition, the DN map $\Lambda$ is a pseudo-differential operator of the first order which coincides with $|\partial_\gamma|$ modulo smoothing operator \cite{LeeU}. Thus $H=-|\partial_\gamma|^{-1}\partial_\gamma$ modulo smoothing operator and both operators (\ref{Hilbert transfrom}) are well-defined on $L^{\mathbb{C}}_2(\Gamma;dl)$. Although the second operator $-\partial^{-1}_\gamma\Lambda$ in (\ref{Hilbert transfrom}) inverts $H$ only on the orthogonal complement of the constants, we keep the (slightly abusing) notation $H^{-1}$ for it. 

Note that $H$ coincides with the standard Hilbert transform on the circle if $M$ is a closed unit disk $\overline{\mathbb{D}}$. The extensions of the standard Hilbert transform on the circle were considered in \cite{BKChar,BSh}, while the above definition (slightly different and based on the connection between $H$ and the reduced rotation $\hat{\Phi}$) is proposed by Belishev.

Denote by ${\rm Tr}$ the trace operator $w\mapsto w|_\Gamma$ on $\Gamma$. In view of the aforementioned, we arrive at the following statement.
\begin{lemma}
\label{traces of holom func lemma}
$H$ is an anti-hermitian operator in the space $(\partial_\gamma H^{3/2}(\Gamma;\mathbb{C}), \, (\cdot,\cdot)_\Lambda)$. The spectrum of $H$ consists of $0$ {\rm(}with ${\rm Ker}H=\mathbb{C}${\rm)}, $\pm i$ {\rm(}with the eigenspaces 
\begin{align*}
{\rm Ker}&(H+i)={\rm clos}_{H^{1/2}(\Gamma;\mathbb{C})}\big(\{\eta\in {\rm Tr}\mathscr{A}(M) \ | \ \langle\eta\rangle=0\}\big), \\
{\rm Ker}&(H-i)={\rm clos}_{H^{1/2}(\Gamma;\mathbb{C})}\big(\{\eta\in {\rm Tr}\overline{\mathscr{A}(M)} \ | \ \langle\eta\rangle=0\}\big),
\end{align*}
respectively{\rm)}, and eigenvalues {\rm (\ref{eigenvalues})}. The eigenfunctions
\begin{equation}
\label{basis of eigenfunctions}
\eta_{\pm k}=\overline{\eta_{\mp k}} \qquad (k=1,\dots,\mathfrak{g})
\end{equation}
corresponding to $\lambda_{\pm k}$ are smooth. {\rm(}In what follows, we assume that eigenfunctions {\rm(\ref{basis of eigenfunctions})} are normalized in $L_2^\mathbb{C}(\Gamma;dl)$ and the eigenfunctions corresponding to the same eigenvalue are orthogonal in $L_2^\mathbb{C}(\Gamma;dl)$.{\rm)}
\end{lemma}
\begin{proof}
It remains to check that $\eta_{\pm k}\in C^{\infty}(\Gamma;\mathbb{C})$. To this end, recall that $H=|\partial_\gamma|^{-1}\partial_\gamma$ and $H^{-1}=\partial_\gamma^{-1}|\partial_\gamma|=-H$ modulo smoothing operators. Thus, the claim follows from the equality $0\ne (\lambda_{\pm k}+\lambda_{\pm k}^{-1})\eta_{\pm k}=(H+H^{-1})\eta_{\pm k}$.
\end{proof}

\paragraph{Double cover.} The double of $(M,g)$ is the surface $(\mathbb{M},{\rm g})$ obtained by gluing $(M,g)$ with its copy (endowed with the opposite orientation) along the boundary. In the subsequent, we consider $M$ to be embedded into $\mathbb{M}$. Introduce the involution $\tau$ on $\mathbb{M}$ which interchanges any point $x$ of $M$ with the same point on its copy. Then $\tau(x)=x$ if and only if $x\in\Gamma$. The projection $\pi: \ \mathbb{M}\mapsto M$ defined by $\pi(x):=\pi(\tau(x)):=x$ for any $x\in M\subset\mathbb{M}$ is continuous, open and discrete, and the set of its ramification points coincides with $\Gamma$. 

The metric ${\rm g}$, the rotation ${\bm \Phi}$, and the Hodge operator ${\bm \star}$ on $\mathbb{M}$ are obtained by gluing together the corresponding metrics, rotations, e.t.c., on $M$ and its copy. By construction, the metric is symmetric 
$$\tau^*{\rm g}={\rm g}=\pi^*g,$$
while the rotation and the Hodge operator are anti-symmetric
\begin{equation}
\label{anticommutaion involution and rotation}
d\tau\circ{\bm\Phi}=-{\bm\Phi}\circ d\tau, \quad \tau^*\circ{\bm \star}=-{\bm \star}\circ\tau^*
\end{equation}
with respect to the involution $\tau$. To check that ${\bm\Phi}$ and ${\bm\star}$ are correctly defined (i.e. they are continuous on the whole $\mathbb{M}$, including $\Gamma$), it is sufficient to note that $d\tau(\gamma)=\gamma$ and $d\tau(\nu)=-\nu$, whence $({\bm \Phi}\nu)|_{\Gamma_-}=-({\bm \Phi}\circ d\tau(\nu))|_{\Gamma_-}=d\tau(({\bm \Phi}\nu)|_{\Gamma_+})=d\tau(\gamma)=\gamma=\Phi\nu=({\bm \Phi}\nu)|_{\Gamma_+}$, where $\Gamma_+$ and $\Gamma_-$ denotes the sides of $\Gamma$ internal and external with respect to $M$, respectively. In particular, $\mathbb{M}$ is orientable. 

Note that, although the metric ${\rm g}$ is, in general, only Lipschitz continuous on $\Gamma$, the rotation and the Hodge operator are smooth on the whole $\mathbb{M}$. Moreover, $\mathbb{M}$ is endowed with the complex structure compatible with the complex structures on $M$ and its copy (i.e., the latter can be considered as complex submanifolds of $\mathbb{M}$). To make sure of this, it is sufficient to construct appropriate holomorphic charts in the neighbourhood of $\Gamma$. Let $x_0$ be an arbitrary point of $\Gamma$ and let $u$ be a smooth harmonic function in $M$ obeying $\partial_\nu u=0$ and $u(x_0)=0$, $\partial_\gamma u(x_0)>0$. In view of the Poincar\'e lemma, we have $\Phi\nabla u=\nabla v$ in some (simple connected) neighborhood $U$ of $x_0$ in $M$. In particular $\partial_\gamma v=\partial_\nu u=0$ in $\Gamma\cap U$ and one can assume that $v=0$ on $\Gamma\cap U$. Then $w=u+iv$ is holomorphic in $U$ and real-valued in $\Gamma\cap U$. Since $x_0$ is a simple zero of $w$ and $\partial_\nu v(x_0)=-\partial_\gamma u(x_0)<0$, one can assume, by decreasing the diameter of $U$, that $w: \ U\mapsto\mathbb{C}_+$ is an injection. Now we extend $w$ on $U\cup\tau(U)$ by symmetry $w\circ\tau=\overline{w}$; then $w: \ U\mapsto\mathbb{C}$ is an injection and $w$ is holomorphic on $\tau(U)$ in view of (\ref{anticommutaion involution and rotation}). Thus, $(U,w)$ is a holomorphic chart on $\mathbb{M}$ which is compatible with the complex atlases of $M$ and its copy $\tau(M)$.

Note that, if the function $w$ is holomorphic in a domain $U\subset\mathbb{M}$, then $w^\dag=\overline{w\circ\tau}$ is holomorphic in $\tau(U)$ due to (\ref{anticommutaion involution and rotation}).

\paragraph{Abelian differentials.} A (complex) 1-form $\omega$ on $\mathbb{M}$ is called an Abelian differential (of the first kind) if the equations
\begin{equation}
\label{Abelian differential def}
i{\bm\star}\omega=\omega, \quad d\omega=0
\end{equation}
hold in $\mathbb{M}$ or, equivalently, if it can be locally (in any simple connected neighborhood) represented as $\omega=dw$, where $w$ is a holomorphic function. The space $H^0(\mathbb{M};K)$ of Abelian differentials of the first kind has the complex dimension ${\rm dim}H^0(\mathbb{M};K)={\rm gen}(\mathbb{M})=2\mathfrak{g}$. 

In view of (\ref{anticommutaion involution and rotation}), the map
$$\omega\mapsto\omega^\dag:=\overline{\tau^*(\omega)}$$
preserves equations (\ref{Abelian differential def}) and therefore it is an involution on $H^0(\mathbb{M};K)$. We call that $\omega\in H^0(\mathbb{M};K)$ is symmetric and write $w\in H^1_{sym}(\mathbb{M};K)$ if $\omega^\dag=\omega$. Then $H^0_{sym}(\mathbb{M};K)$ is a real linear space of dimension $\mathfrak{g}$ and any $\omega\in H^0(\mathbb{M};K)$ admits the decomposition $\omega=\omega_++i\omega_-$, where $\omega_+=(\omega+\omega^\dag)/2$ and $\omega_-=(\omega-\omega^\dag)/2i$ belong to $H^0_{sym}(\mathbb{M};K)$.

The important observation used in the paper is the following connection between the tangent harmonic fields on $M$ and the symmetric Abelian differentials on its double $\mathbb{M}$. 
\begin{lemma}
\label{Connection between normal fields and symmetric differentials}
$A\in\mathcal{D}$ if and only if $(A+i\Phi A)^\flat$ is a restriction on $M$ of a symmetric Abelian differential $\omega_A$ on $\mathbb{M}$. The map $A\mapsto\omega_A$ is a bijection from $\mathcal{D}$ onto $H^1_{sym}(\mathbb{M};K)$.
\end{lemma}
\begin{proof}
Let $A\in\mathcal{D}$ and let $\omega_A$ be the $1$-form on $\mathbb{M}$ given by $\omega_A:=(A+i\Phi A)^\flat$ on $M$ and extended to $\tau(M)$ by symmetry $\omega_A^\dag=\omega_A$. Let $U$ be a simple connected neighborhood in $M$; since $A$ and $\Phi A$ are harmonic, they can be represented as $A=\nabla u$, $\Phi A=\nabla v$ in $U$ and the function $w=u+iv$ is holomorphic in $U$. Then $\omega_A=(\nabla w)^\flat=dw$ in $U$. By symmetry $\omega_A^\dag=\omega_A$, we have $\omega_A=\overline{\tau^*dw}=d\overline{w\circ\tau}=dw^\dag$ in $\tau(U)$, where $w^\dag$ is holomorphic in $\tau(U)$. If $\Gamma\cap U$ is a segment of non-zero length, then $0=A_\nu=\partial_\nu u=\partial_\gamma v$ and one can chose $v$ in such a way that $v=0$ on $U\cap\Gamma$. Then $w|_{\Gamma_+}(x)=u|_{\Gamma_+}(x)=u\circ\tau|_{\Gamma_+}(x)=u|_{\Gamma_-}(x)=w^\dag|_{\Gamma_-}(x)$ for $x\in\Gamma\cap U$ and, due to the Schwarz reflection principle, $w$ admits holomorphic extension (still denoted by $w$) on $U\cup\tau(U)$ which coincides with $w^\dag$ on $\tau(U)$. Therefore $\omega$ admits the representation $\omega=dw$ with holomorphic $w$ in any simple connected neighborhood in ${\rm M}$ and, hence, $\omega\in H^1_{sym}(\mathbb{M};K)$. The map $A\mapsto\omega_A$ is an injection due to the uniqueness of the analytic continuation. 

Now, suppose that $\omega\in H^1_{sym}(\mathbb{M};K)$ and $\omega^\sharp=A+iB$. Then $A$,$B$ are harmonic since ${\rm div}(A+iB)={\bm\star}d{\bm\star}\omega=-i{\bm\star}d\omega=0$ and ${\rm div}(\Phi(A+iB))=-{\bm\star}d\omega=0$. An addition, $A+iB=\omega^\sharp=(i{\bm\star}\omega)^\sharp=i(\Phi A+i\Phi B)=-\Phi B+i\Phi A$, whence $B=\Phi A$ and $\omega=(A+i\Phi A)^\flat$. Finally, $\omega(\nu)=\omega^\dag(\nu)=\overline{\omega(d\tau(\nu))}=-\overline{\omega(-\nu)}$, whence $A_\nu=\Re\omega(\nu)=-\Re\omega(\nu)=0$. Therefore, $A\in\mathcal{D}$ and $\omega=\omega_A$. This means that the map $A\mapsto\omega_A$ is a surjection. As a corollary, we have ${\rm dim}\mathcal{D}={\rm dim}H^1_{sym}(\mathbb{M};K)=g$ which explains formula (\ref{dim of defect space}).
\end{proof}

\paragraph{Homology groups.} Let $(X,g)$ be an oriented surface (possibly with non-empty boundary) of genus $m$ and let $l$ be a finite (possibly empty) collection of closed oriented curves in $M$. By definition, the following operations preserve homology class (`cycle') $[l]$ of $l$: a) a homotopic deformation of each curve in $l$, b) cutting the curves into a finite number of segments and gluing them together in a different order in such a way that the resulting curves are closed and the orientation of each segment is preserved, c) adding or excluding an oriented boundary of some (arbitrarily oriented) domain in $X$. The set of cycles endowed with the addition $[l]+[l']=[l\cup l']$ is an Abelian group $H_1(X,\mathbb{Z})$ called the first homology group. Note $-[l]=[-l]$, where $-l$ is obtained from $l$ by reversing the orientation of all curves. It is well known that $H_1(X,\mathbb{Z})\simeq \pi_1(X^\circ)/[\pi_1(X^\circ),\pi_1(X^\circ)]\simeq\mathbb{Z}^{2m}$, where $X^\circ$ is obtain from $X$ by attaching disks to all connected components of $\partial X$. 

Let $l,l'$ be closed oriented curves in $X$; by homotopic deformation one can assume that they are smooth, oriented by unit tangent vectors $\gamma,\gamma'$, respectively, $l$ intersect $l'$ a finite number of times, and each intersection is transversal. The intersection is positive if $(\gamma,\gamma')$ is positively oriented with respect to the orientation of $X$, and negative otherwise. By definition, the {\it intersection number} $[l]\sharp[l']$ is the difference between numbers of positive and negative intersections of $l$ and $l'$ (if $l,l'$ are collections of the curves, then $[l]\sharp[l']$ is obtained by the summation of the intersection numbers of all pairs from $l\times l'$). It can be shown that $[l]\sharp[l']$ is invariant with respect to operations a)-c) and thereby is well defined on homology classes. Moreover, $\sharp: \ H_1(X,\mathbb{Z})\times H_1(X,\mathbb{Z})\mapsto\mathbb{Z}$ is an alternating bilinear form.

We say that $[l_\cdot]=\{[l_1],\dots,[l_{2m}]\}$ form a {\it homology basis} on $X$ if they generate $H_1(X,\mathbb{Z})$. Introduce the intersection matrix $J$ of the basis $[l_\cdot]$ by $J_{ij}:=[l_i]\sharp\,[l_j]$. The homology basis is called {\it canonical} if its intersection matrix coincides with the standard symplectic matrix
$$\Omega_{(m)}=\left(\begin{array}{cc} 0 & I_{m}\\ -I_{m} & 0\end{array}\right).$$
In this case we call that $a_1=[l_1],\dots,a_m=[l_m]$ are $a$-cycles and $b_1=[l_{m+1}],\dots,b_m=[l_{2m}]$ are $b$-cycles. The canonical bases always exist. Two homology bases $[l_\cdot]$, $[l'_\cdot]$ are simultaneously (non-)canonical if and only if $[l'_i]=\sum_j M_{ij}[l_j]$ ($i=1,\dots,m$), where $M\in{\rm Sp}(m,\mathbb{Z})$. A compact Riemann surface $X$ with empty boundary endowed with a choice of a canonical homology basis $[l_\cdot]$ is called a {\it Torelli marked surface}.

Let $\partial X$ be empty or diffeomorphic to a circle, let $\omega$ be a harmonic 1-form on $X$, and let $A=\omega^\sharp$. The integral
$$T(\omega|[l])\equiv T(A|[l]):=\int_{l}\omega=\int_{l}g(A,\gamma)dl$$
(where $\gamma$ and $dl$ are tangent unit vector and the length element on $l$, respectively) depends only on $[l]$; this integral is called the {\it period} of $\omega$ (or of $A$) along the cycle $[l]$. 

We say that $\omega$ is normal (tangent) to $\partial X$ if $\omega(\gamma)=0$ ($\omega(\nu)=0$) on $\partial X$, or, equivalently, if $A=\omega^\sharp$ is normal (tangent) to $\partial X$. Then harmonic 1-forms $\omega$ normal (tangent) to $\partial X$ are determined by their period vectors
$$\mathrm{T}(\omega|[l_\cdot])\equiv\mathrm{T}(\omega|[l_\cdot]):=(T(\omega|[l_1],\dots,T(\omega|[l_m])^T$$
with respect to a given homology basis $[l_\cdot]$.
\begin{lemma}
\label{inner product via periods lemma}
Let $(X,g)$ be an orientable surface of genus $m$ {\rm (}possibly with non-empty boundary{\rm)}, let $\Phi$,$\star$ be the rotation and the Hodge operator on $X$, respectively, and let $[l_\cdot]$ be a homology basis on $X$. Let $A,B\in C^\infty(X;TX)$ satisfy ${\rm div}A={\rm div}(\Phi B)=0$ in $X$, and let $A$ be tangent or $B$ is normal to $\partial X$. Then their inner product in $L_2(X;TX)$ admits the representation 
\begin{align}
\label{inner product via periods}
(B,A)=\mathrm{T}(B|\,[l_\cdot])^T J^{-1}\mathrm{T}(-\Phi A|\,[l_\cdot]),
\end{align}
where $J$ is the intersection matrix of $[l_\cdot]$.
\end{lemma}
\begin{proof}
Denote $\omega=B^\flat$ and $\eta={\star}A^\flat$, then $d\omega=d\eta=0$ and at least one of $A,B$ is normal to $\partial X$. The left-hand side of (\ref{inner product via periods}) can be rewritten as $(B,A)=\int_X \omega^\wedge\eta$ while the right-hand side is given by $\mathrm{T}(\omega|\,[l_\cdot])^T(-J^{-1})\mathrm{T}(\eta|\,[l_\cdot])$. Let us show that the right-hand side is independent of the choice of a homology basis. Let $[l'_\cdot]$ be a new homology basis connected with ${l_\cdot}$ via $[l'_i]=\sum_{ij}M_{ij}[l_j]$ (i.e., $M,M^{-1}$ have integer entries). Then the period vectors and the intersection matrices obey the transformation rules ${\rm T}(\cdot|[l'_\cdot])=M{\rm T}(\cdot|[l_\cdot])$ and 
\begin{equation}
\label{transformation rule for intersection matrix}
J'=MJM^T,
\end{equation}
whence $\mathrm{T}(\omega|\,[l'_\cdot])^T(-J^{'-1})\mathrm{T}(\eta|\,[l'_\cdot])=\mathrm{T}(\omega|\,[l_\cdot])^T(-J^{-1})\mathrm{T}(\eta|\,[l_\cdot])$. Thus, one can check (\ref{inner product via periods}) assuming that $[l_\cdot]=\{a_1,\dots,a_m,b_1,\dots,b_m\}$ is canonical. Then (\ref{inner product via periods}) takes the familiar form
\begin{equation}
\label{Riemann bilinear identity}
\int_X \omega^\wedge\eta=\sum_{j=1}^m\Big(\int_{a_j}\omega\int_{b_j}\eta-\int_{b_j}\omega\int_{a_j}\eta\Big),
\end{equation}
which is just the Riemann bilinear identity if $\partial X=\varnothing$. It remains to prove that (\ref{Riemann bilinear identity}) remains valid if $\partial X\ne\varnothing$ and one of $\omega,\eta$ is normal to $\partial X$. Let $X^\circ$ be Riemann surface obtained by attaching a disk $\mathcal{D}=\{z\in\mathbb{C} \ | \ |z|\le 1\}$ to each connected component of $\partial X$ (to construct the complex charts near $\partial X\subset X^\circ$, one can use the procedure described after (\ref{anticommutaion involution and rotation})). Let $\omega_\circ,\eta_\circ$ be smooth extensions to $X^\circ$ of $\omega,\eta$, respectively, given by $\omega^\circ=du_1$, $\eta^\circ=du_2$, where $u_k$ are smooth on $\tilde{X}=X^\circ\backslash{\rm int}X$. Denote by $\chi$ the smooth function with compact support on $[0,+\infty)$ equal to $1$ in the neighborhood of zero. Introduce the function $\chi_\varepsilon$ given by $\chi_\varepsilon(x)=\chi(\varepsilon^{-1}(|z(x)|-1))$ on each disk in $\tilde{X}$. Suppose that $\omega$ is normal to $\partial X$; then one can chose $u_1$ in such a way that $u_1=0$ on $\partial X$. Let $\omega_\varepsilon$ be the smooth closed extension of $\omega$ given by $\omega_\varepsilon=d(\chi_\varepsilon u_1)$ on $\tilde{X}$. Since $u(z)=O(1-|z|)=O(\varepsilon)$ on the support of $\chi_\varepsilon$, we have $\|\omega_\varepsilon\|_{L_2(\tilde{X}\,;T^*\tilde{X})}=O(\varepsilon^{1/2})$, whence 
\begin{equation}
\label{mollifier}
\int_{X_\circ} \omega_\varepsilon^\wedge\eta_\circ\to \int_X \omega^\wedge\eta \qquad (\varepsilon\to 0).
\end{equation}
Since each closed curve in $X^\circ$ can be homotopically deformed to a curve in $X\subset X^\circ$, we have $H_1(X^\circ,\mathbb{Z})=H_1(X,\mathbb{Z})$ and each homology class on $X^\circ$ is an extension of a homology class on $X$. Thus, formula (\ref{Riemann bilinear identity}) is valid with the left-hand side replaced by the left-hand side of $(\ref{mollifier})$. Now, formula (\ref{Riemann bilinear identity}) is obtained by the limit transition as $\varepsilon\to 0$.
\end{proof}
As easily follows from Lemma \ref{inner product via periods lemma} and (\ref{Abelian differential def}), Abelian differentials are determined by their $a$-periods.

Let $X=\mathbb{M}$. The involution $\dag$ acting on curves in $\mathbb{M}$ by the rule $l^\dag:=\tau\circ l$ induce the involution $\dag$ on $H_1(\mathbb{M},\mathbb{Z})$ obeying $[l]^\dag=[l^\dag]$. Since the involution $\tau$ is orientation reversing, we have
\begin{equation}
\label{involution intersection}
[l]^\dag\sharp[l']^\dag=-[l]\sharp[l'] \qquad ([l],[l']\in H_1(\mathbb{M},\mathbb{Z})).
\end{equation}
Note that
\begin{equation}
\label{involution of periods}
T(\omega^\dag|\,[l]^\dag)=\int_{\tau\circ l}\overline{\tau^*\omega}=\overline{\int_{l}\omega}=\overline{T(\omega|\,[l])} \qquad (\omega\in H^0(\mathbb{M};k)).
\end{equation}
Since $\partial M=\Gamma$ consists of one connected component, each homology class $[l]$ in $\mathbb{M}$ admits the decomposition $[l]=[l_+]+[l_-]^\dag$, where $l_\pm$ are collections of the curves in $M$. In particular, we have $H_1(\mathbb{M},\mathbb{Z})=H_1(M,\mathbb{Z})+H_1(M,\mathbb{Z})^\dag\simeq 2H_1(M,\mathbb{Z})$. Due to this facts and (\ref{involution intersection}), any homology basis $a_1,\dots,a_{\mathfrak{g}},b_1,\dots,b_{\mathfrak{g}}$ in $H_1(M,\mathbb{Z})$ defines the canonical homology basis
\begin{equation}
\label{symmetric canonical basis}
a_1,\dots,a_{\mathfrak{g}}, \ a_{\mathfrak{g}+1}:=a^\dag_1,\dots,a_{2\mathfrak{g}}:=a^\dag_{\mathfrak{g}}, \ b_1,\dots,b_{\mathfrak{g}}, \ b_{\mathfrak{g}+1}:=-b_1^\dag,\dots,b_{2\mathfrak{g}}:=-b^\dag_{\mathfrak{g}}.
\end{equation}
in $H_1(\mathbb{M},\mathbb{Z})$. In what follows, a homology basis of the form (\ref{symmetric canonical basis}) is called {\it symmetric}.

\paragraph{Period matrices.} Consider a Torelli marked Riemann surface $(X,[l_\cdot])$ of genus $m$ (here $[l_\cdot]=\{a_1,\dots,a_{m},b_1,\dots,b_{m}\}$). For a basis $\omega_\cdot=\{\omega_1,\dots,\omega_{\mathrm{g}}\}$ in $H^0(X;K)$, we introduce its {\it period matrix} $\mathbb{T}([l_\cdot],\omega_\cdot)$ with entries $\mathbb{T}_{ij}([l_\cdot],\omega_\cdot):=T(\omega_i|[l_j])$. There is the unique basis $\omega_\cdot$ whose period matrix is of the form 
$$\mathbb{T}([l_\cdot],\omega_\cdot)=(I_{m}|\mathbb{B});$$ 
this basis is called {\it dual} to $[l_\cdot]$ and the matrix $\mathbb{B}$ is called the $b$-{\it period matrix} of $(X,[l_\cdot])$. We say that a basis $\omega_\cdot$ in $H^0(X;K)$ is {\it canonical} if it is dual to some Torelli marking on $X$. Also we say that a matrix $\mathbb{B}$ is $b$-period matrix of $X$ if there is a Torelli marking $[l_\cdot]$ on $X$ such that $\mathbb{B}$ is a $b$-period matrix of $(X,[l_\cdot])$.

Let $X=\mathbb{M}$ and let $\omega_\cdot=\{\omega_1,\dots,\omega_{2\mathrm{g}}\}$ be a basis in $H^0(\mathbb{M};K)$. The basis $\omega_\cdot$ is called symmetric canonical if it is dual to some symmetric canonical homology basis $[l_\cdot]$ on $\mathbb{M}$. In this case, the $b$-period matrix $\mathbb{B}$ of $(X,[l_\cdot])$ is called symmetric.

Now, let $[l_\cdot]=\{[l_1],\dots,[l_{2\mathfrak{g}}]\}$ be a homology basis on $M$ and let $B_\cdot=\{B_1,\dots,B_{2\mathfrak{g}}\}$ be a basis in $\mathcal{N}$. We say that $B_\cdot$ is {\it dual} to $[l_\cdot]$ if $T(B_i|[l_j])=\delta_{ij}$; in this case, the matrix $\mathfrak{B}$ with the entries 
$$\mathfrak{B}_{ji}=T(\Phi B_i|[l_j])$$
is called the {\it auxiliary period matrix} corresponding to the homology basis $[l_\cdot]$. The basis dual to $[l_\cdot]$ exists and unique. Indeed, if $B'_\cdot=\{B'_1,\dots,B'_{2\mathfrak{g}}\}$ is a basis in $\mathcal{N}$, then the matrix $M$ with entries $M_{ij}=T(B_i|[l_j])$ is invertible (otherwise, there is the vector $0\ne B=\sum_i c_i B_i\in\mathcal{D}$ which is harmonic, normal to $\Gamma$ and has periods $T(B_i|[l_j])=\sum_i c_iM_{ij}=0$, a contradiction). Thus, the dual basis to $[l_\cdot]$ is composed of the vectors $B_i=\sum_j (M^{-1})_{ij}B'_j$.

We say that the basis $B_\cdot$ in $\mathcal{N}$ is dual if it is dual to some homology basis $[l_\cdot]$ on $M$; if, in addition, $[l_\cdot]$ is canonical, then we say that $B_\cdot$ is canonical. As follows from (\ref{transformation rule for intersection matrix}), two dual bases $B_\cdot$ and $B'_\cdot$ are simultaneously (non-)canonical if and only if $B'_i=\sum_j M_{ji}B_j$, where $M\in{\rm Sp}(\mathfrak{g},\mathbb{Z})$ is arbitrary; the corresponding homology bases are connected via $[l'_i]=\sum_j(M^{-1})_{ij}[l_j]$. Similarly, any two auxiliary period matrices $\mathfrak{B}$ and $\mathfrak{B}'$ (corresponding to different homology bases) are related via
\begin{equation}
\label{aux period matrix transformation rule}
\mathfrak{B}'=M^{-1}\mathfrak{B}M \qquad (M\in {\rm Sp}(\mathfrak{g},\mathbb{Z})).
\end{equation}

The following lemma provides the criterion of the canonicity of the dual basis. Also, it provides the expression for the auxiliary period matrix of a canonical homology basis in terms of inner products between elements of its dual basis.
\begin{lemma}
\label{canonicity criterion lemma}
\begin{enumerate}[label=\alph*{\rm)}]
\item The intersection matrix $J$ of the homology basis $[l_\cdot]=\{[l_1],\dots,[l_{2\mathfrak{g}}]\}$ in $M$ can expressed in terms of its dual basis $B_\cdot=\{B_1,\dots,B_{2\mathfrak{g}}\}$ as
\begin{equation}
\label{intersection form via alternating form}
(J^{-1})_{ij}=(B_i,\Phi B_j).
\end{equation}
\item The dual basis $B_\cdot$ is canonical if and only if $(\Phi B_i,B_j)=(\Omega_{(\mathfrak{g})})_{ij}$ for all $i,j=1,\dots,\mathfrak{g}$.
\item If $[l_\cdot]$ is a canonical homology basis in $M$, then its auxiliary period matrix can be expressed in terms of its dual basis $B_\cdot=\{B_1,\dots,B_{2\mathfrak{g}}\}$ as
$$(\Omega_{(\mathfrak{g})}\mathfrak{B})_{ij}=(B_i,B_j).$$
\end{enumerate}
\end{lemma}
\begin{proof}
{\it a}) Since $B_k$ are normal and $T(B_k|[l_s])=\delta_{ks}$, Lemma \ref{inner product via periods lemma} and the equality $\Phi^2=-{\rm Id}$ imply
$$(B_i,\Phi B_j)=\mathrm{T}(B_i|\,[l_\cdot])^T J^{-1}\mathrm{T}(-\Phi^2 B_j|\,[l_\cdot])=\delta_{ik}(J^{-1})_{ks}\delta_{sj}=(J^{-1})_{ij}.$$
Thus, we have proved (\ref{intersection form via alternating form}). Now {\it b}) easily follows from {\it a}). {\it c}) In view of Lemma \ref{inner product via periods lemma} and the equality $J^{-1}=\Omega_{(\mathfrak{g})}^{-1}=-\Omega_{(\mathfrak{g})}$, we have
$$(B_i,B_j)=\mathrm{T}(B_i|\,[l_\cdot])^T\Omega_{(\mathfrak{g})}\mathrm{T}(\Phi B_j|\,[l_\cdot])=\delta_{ik}(\Omega_{(\mathfrak{g})})_{ks}\mathfrak{B}_{sj}=(\Omega_{(\mathfrak{g})}\mathfrak{B})_{ij}.$$
\end{proof}

Let $[l_\cdot]=\{a_1,\dots,a_{\mathfrak{g}},b_1,\dots,b_{\mathfrak{g}}\}$ be a canonical homology basis on $M$, let $B_\cdot$ be the corresponding dual basis in $\mathcal{N}$ and let $\mathfrak{B}$ be the corresponding auxiliary period matrix. Let us establish the connection between $\mathfrak{B}$ and the $b$-period matrix $\mathbb{B}$ of the cover $\mathbb{M}$ corresponding to the symmetric canonical basis $[\tilde{l}_\cdot]$ related to $[l_\cdot]$ via (\ref{symmetric canonical basis}). Denote
$$\omega_i=(iB_i-\Phi B_i)^\flat \qquad (i=1,\dots,2\mathfrak{g}).$$
As follows from Lemma \ref{Connection between normal fields and symmetric differentials} and the equality $\Phi\mathcal{N}=\mathcal{D}$, $\omega_i$ admit analytic continuation to symmetric abelian differentials on the double $\mathbb{M}$ (still denoted by $\omega_i=\omega_i^\dag$). Note that $\omega_1,\dots,\omega_{2\mathfrak{g}}$ constitute a basis in $H^1(\mathbb{M};K)$ due to the linear independence of $B_1,\dots,B_{2\mathfrak{g}}$. In addition,
$$T(\omega_i|a_j)=i\delta_{ij}-\mathfrak{B}_{ji}, \qquad T(\omega_i|b_j)=i\delta_{i,j+\mathfrak{g}}-\mathfrak{B}_{j+\mathfrak{g},i}$$
for $j\le\mathfrak{g}$. In view of (\ref{involution of periods}) and (\ref{symmetric canonical basis}), we have
\begin{align*}
T(\omega_i|a_j)&=T(\omega_i^\dag|a^\dag_{j-\mathfrak{g}})=\overline{T(\omega_i|a_{j-\mathfrak{g}})}=-i\delta_{i,j-\mathfrak{g}}-\mathfrak{B}_{j-\mathfrak{g},i},\\
T(\omega_i|b_j)&=T(\omega_i^\dag|-b^\dag_{j-\mathfrak{g}})=-\overline{T(\omega_i^\dag|b^\dag_{j-\mathfrak{g}})}=i\delta_{i,j}-\mathfrak{B}_{j,i}
\end{align*}
for $j>\mathfrak{g}$. Then the period matrix of the basis $\omega_\cdot$ is given by
$$\mathbb{T}([\tilde{l}_\cdot]|\omega_\cdot)=\big(i\chi^{+-}-\mathfrak{B}^T\chi^{++}|i\chi_{++}-\mathfrak{B}^T\chi_{+-}\big),$$
where 
\begin{eqnarray}
\label{matrices}
\chi^{\mathfrak{s},\mathfrak{s}'}=\left(\begin{array}{cc}
\mathfrak{s}I_{\mathfrak{g}} & \mathfrak{s}'I_{\mathfrak{g}}\\
0 & 0
\end{array}\right), \qquad \chi_{\mathfrak{s},\mathfrak{s}'}=\left(\begin{array}{cc}
0 & 0 \\
\mathfrak{s}I_{\mathfrak{g}} & \mathfrak{s}'I_{\mathfrak{g}}
\end{array}\right) \qquad (\mathfrak{s},\mathfrak{s}'=\pm)
\end{eqnarray}
and $I_{m}$ is the unit $m\times m$-matrix. Introduce the new basis $\tilde{\omega}_\cdot$ in $H^1(\mathbb{M};K)$ by $\tilde{\omega}_i=\sum_j M_{ij}\omega_j$, where $M=(i\chi^{+-}-\mathfrak{B}^T\chi^{++})^{-1}$. Then its period matrix is equal to
$$\mathbb{T}([\tilde{l}_\cdot]|\tilde{\omega}_\cdot)=(I_{2\mathfrak{g}} \ | \ (i\chi^{+-}-\mathfrak{B}^T\chi^{++})^{-1}(i\chi_{++}-\mathfrak{B}^T\chi_{+-})\big).$$
Hence, the basis $\tilde{\omega}_\cdot$ is dual to $[\tilde{l}_\cdot]$. In particular, the $b$-period matrix $\mathbb{B}$ of $\mathbb{M}$ corresponding to $[\tilde{l}_\cdot]$ is related to $\mathfrak{B}$ via
\begin{equation}
\label{connection of period matrices}
\mathbb{B}=(i\chi^{+-}-\mathfrak{B}^T\chi^{++})^{-1}(i\chi_{++}-\mathfrak{B}^T\chi_{+-}).
\end{equation}
Symmetry (\ref{symmetric canonical basis}) of the canonical homology basis leads to the symmetries of the dual basis $\omega_\cdot$ and the $b$-period matrix. Indeed, (\ref{symmetric canonical basis}) and (\ref{involution of periods}) imply
$$T(\omega_{i+\mathfrak{g}}|a_{j+\mathfrak{g}})=\delta_{ij}=T(\omega_i|a_j)=T(\omega^\dag_i|a_{j+\mathfrak{g}}), \qquad T(\omega_{i+\mathfrak{g}}|a_{j})=0=T(\omega_i|a_{j+\mathfrak{g}})=T(\omega^\dag_i|a_{j})$$
for $j\le\mathfrak{g}$. Since the Abelian differentials are determined by their $a$-periods, we have $\omega_{g+i}=\omega^\dag_i$ for $i=1,\dots,\mathfrak{g}$. As a corollary, we obtain
\begin{equation}
\label{period matrix symmetries}
\mathbb{B}_{\mathfrak{g}+i,\mathfrak{g}+j}=T(\omega^\dag_i|,-b^\dag_{j})=-\overline{\mathbb{B}_{ij}}, \qquad \mathbb{B}_{\mathfrak{g}+i,j}=T(\omega^\dag_i|b_j)=\overline{T(\omega_i|-b_{\mathfrak{g}+j})}=-\overline{\mathbb{B}_{i,\mathfrak{g}+j}}.
\end{equation}

\section{Procedure for determination of period matrix of double cover of $M$ from its DN map}
\label{Algorithm sec}
\paragraph{Step 1. Determination of boundary data of harmonic normal vectors on $M$.} Let $u=u^f$ be a harmonic function in $M$ with trace $f$ on $\Gamma$. Then $\Phi\nabla u$ is a harmonic field and the decomposition $\mathcal{H}=\mathcal{E}\oplus\mathcal{D}$ yields 
\begin{equation}
\label{Helmholtz decomposition}
\Phi\nabla u^f=\nabla u^h+A,
\end{equation}
where $u^h$ is a harmonic function in $M$ with trace $h$ on $\Gamma$ and $A\in\mathcal{D}$. Note that $A=(\Phi-\hat{\Phi})\nabla u^f$ and $\nabla u^h=\hat{\Phi}\nabla u^f$. 

The vector field 
\begin{equation}
\label{normal field B}
B=\Phi A
\end{equation}
is an element of $\mathcal{N}$. Restricting equations (\ref{Helmholtz decomposition}), (\ref{normal field B}) to the boundary and taking into account (\ref{orientation of boundary}), we obtain
\begin{equation}
\label{Helmholtz decomposition on boundary}
-\partial_\gamma f=\Lambda h, \qquad \Lambda f=\partial_\gamma h+A_\gamma, \qquad B_\nu=-A_\gamma.
\end{equation}
In particular, we have
\begin{equation}
\label{B from f}
B_\nu=\partial_\gamma h-\Lambda f=-(\partial_\gamma\Lambda^{-1}\partial_\gamma+\Lambda)f=-\partial_\gamma(H+H^{-1})f.
\end{equation}
As is easily seen from (\ref{B from f}), $f$ is determined by $B$ up to an element of $\Re{\rm Ker}(H+H^{-1})={\rm clos}_{H^{1/2}(\Gamma;\mathbb{R})}\big(\Re{\rm Tr}\mathscr{A}(M)\big)$. This is related to the fact that the fields $A,B$ do not change after addition the Cauchy-Riemann equation $\Phi\nabla u^{\tilde{f}}=\nabla u^{\tilde{h}}$ to (\ref{Helmholtz decomposition}), where $u^{\tilde{f}}+iu^{\tilde{h}}\in{\rm clos}_{H^1(M,g)}\big(\mathscr{A}(M)\big)$. One can fix $f$ by the additional condition 
$$(f,\tilde{f})_\Lambda=0 \qquad \forall \ {\rm Ker}(H+H^{-1});$$
then $f$ is uniquely defined by $B_\nu$ and admits the representation
\begin{equation}
\label{f via eigenfunctions}
f=\sum_{\pm k=1}^{\mathfrak{g}}\hat{f}_k\eta_k \qquad (\hat{f}_{-k}=\overline{\hat{f}_k}\in\mathbb{C}),
\end{equation}
where $\eta_k$ are given by (\ref{basis of eigenfunctions}). In particular, one can assume that $f$ is smooth.

In what follows, we say that $(B_\nu,f)$ is a {\it boundary data} for $B\in\mathcal{N}$ and denote $(B_\nu,f)=\mathfrak{T}(B)$. Note that each $B\in\mathcal{N}$ admits boundary data. Indeed, formula (\ref{all defect vectors}) implies that each $A=-\Phi B$ admits representation (\ref{Helmholtz decomposition}). The space $\mathfrak{T}(\mathcal{N})$ of all boundary data is denoted by $\mathcal{N}_\Gamma$. Since each $B\in\mathcal{N}$ is determined by the normal component of its boundary trace, the linear map $\mathfrak{T}: \ \mathcal{N}\to \mathcal{N}_\Gamma$ is a bijection.

From (\ref{Helmholtz decomposition}), (\ref{normal field B}), (\ref{Green formula}), and (\ref{Helmholtz decomposition on boundary}) it follows that
\begin{equation}
\label{norm of B in terms of boundary data}
\begin{split}
\|B\|^2=&\|A\|^2=\|\Phi\nabla u^f\|^2-\|\nabla u^h\|^2=\|\nabla u^f\|^2-\|\nabla u^h\|^2=\\
=&(\Lambda f,f)_\Gamma-(\Lambda h,h)_\Gamma=(\Lambda f,f)_\Gamma-(\partial_\gamma f,\Lambda^{-1}\partial_\gamma f)_\Gamma=-(B_\nu,f)_\Gamma.
\end{split}
\end{equation}
Due to (\ref{norm of B in terms of boundary data}) and the polarization identity, the inner products between elements of $\mathcal{N}$ can be founded from their boundary data. Namely, we have
$$(B,B')=-(B_\nu,f')_\Gamma=-(f,B'_\nu)_\Gamma,$$
where $(B'_\nu,f')$ is the boundary data of $B'\in\mathcal{N}$. Similarly, since the subspaces $\mathcal{D}$ and $\mathcal{N}=\Phi\mathcal{D}$ are $L_2(M;TM)$-orthogonal to $\mathcal{E}$ and $\Phi\mathcal{E}$, respectively, we have
\begin{align*}
(\Phi B,B')=-(B,\Phi B')=-(\Phi &A,\Phi B')=-(A,B')=(\nabla u^h-\Phi\nabla u^f,B')=\\
=(\Phi&\nabla u^f,B')+(\nabla v^h,B')=0+(\nabla u^h,B')=\\
=&\int_M {\rm div}(u^hB')dS-\int_M u^h{\rm div}(B')dS=\\&=
\int_\Gamma h B'_\nu dl+0=-(\Lambda^{-1}\partial_\gamma f,B'_\nu)_\Gamma=-(Hf,B'_\nu)_\Gamma.
\end{align*}

We arrive at the following statement.
\begin{prop}
\label{determination of normal harmonic fields}
Using the DN map $\Lambda$ of $M$, one can construct the isometric copy $\mathcal{N}_\Gamma$ of $\mathcal{N}$ in the following way:
\begin{itemize}
\item The space $\mathcal{N}_\Gamma$ is defined by 
$$\mathcal{N}_\Gamma:=\Big\{(B_\nu,f) \quad \ | \ f=\sum_{\pm k=1}^{\mathfrak{g}}\hat{f}_k\eta_k, \quad \hat{f}_{-k}=\overline{\hat{f}_k}\in\mathbb{C}, \quad B_\nu=-\partial_\gamma(H+H^{-1})f\Big\},$$
where $H=\Lambda^{-1}\partial_\gamma$ is the Hilbert map and $\eta_{\pm 1},\dots,\eta_{\pm\mathfrak{g}}$ are eigenfunctions {\rm(}\ref{basis of eigenfunctions} {\rm )} corresponding to the eigenvalues $\lambda_{\pm 1},\dots,\lambda_{\pm\mathfrak{g}}$ different from $\pm i$.
\item $\mathcal{N}_\Gamma$ is endowed with the inner product
\begin{equation}
\label{inner product on boundary data}
{\bm(}(B_\nu,f),(B'_\nu,f'){\bm)}:=-(B_\nu,f')_\Gamma=-(f,B'_\nu)_\Gamma
\end{equation}
and the alternating bilinear form
\begin{equation}
\label{alterneting form}
{\bm\langle}(B_\nu,f),(B'_\nu,f'){\bm\rangle}:=(B_\nu,Hf')_\Gamma=-(Hf,B'_\nu)_\Gamma.
\end{equation}
\end{itemize}
Then the map $\mathfrak{T}: \ \mathcal{N}\to \mathcal{N}_\Gamma$ introduced after {\rm(}\ref{f via eigenfunctions}{\rm)} is an isometry obeying
\begin{equation}
\label{alternating form isomorphism}
(\Phi B,B')={\bm\langle}\mathfrak{T}(B),\mathfrak{T}(B'){\bm\rangle} \qquad (B,B'\in\mathcal{N}).
\end{equation}
\end{prop}

\paragraph{Step 2. Determination of boundary data of harmonic normal vectors with integer periods on $M$.} Let us rewrite (\ref{Helmholtz decomposition}), (\ref{normal field B}) as follows
\begin{equation}
\label{Helmholtz decomposition 2}
\Phi\nabla u^h=B-\nabla u^f.
\end{equation}
Let $U$ be an arbitrary simple connected neighborhood in $M$. Since $u^h$ is harmonic in $U$, the Poincar\'e lemma implies that there is a harmonic function $V$ in $U$ obeying $\nabla V=\Phi\nabla u^h$. Hence, the function
\begin{equation}
\label{multivalued function}
x\mapsto W(x):=u^h(x)+iV(x)=u^h(x)+i\int_\cdot^x(\Phi\nabla u^h)^\flat+i{\rm const}
\end{equation} 
is holomorphic in $U$. However $W$ is not in general globally defined on $M$: after analytic continuation along the loop $l$ from any non-trivial cycle $[l]\in H_1(M,\mathbb{Z})$ in $M$ its value acquires the shift 
$$T(\Phi\nabla u^h|[l])=T(B|[l])$$
(the equality follows from (\ref{Helmholtz decomposition 2})). Note that one can chose a single-valued branch of $W$ in a tubular neighborhood of $\Gamma$ due to $\int_\Gamma B^\flat=\int_\Gamma B_\gamma dl=0$. In view of (\ref{Helmholtz decomposition 2}) and (\ref{Helmholtz decomposition on boundary}), the boundary trace of $W$ obeys
\begin{equation*}
\begin{split}
\partial_\gamma W|_\Gamma=\partial_\gamma h+ig(\nabla V,\gamma)=\partial_\gamma h+ig(\Phi\nabla u^h,\gamma)=\partial_\gamma h+i\partial_\nu u^h=\\=(\partial_\gamma+i\Lambda) h=-(\partial_\gamma+i\Lambda)\Lambda^{-1}\partial_\gamma f=-\partial_\gamma(H+i)f.
\end{split}
\end{equation*}
Hence,
$$W|_\Gamma=-(H+i)f+iC$$
Here $C\in \mathbb{R}$ is a constant on $\Gamma$ which depends on the choices of the constant in (\ref{multivalued function}) and branch of $W$ near $\Gamma$. From now on, we assume that (the branch of) $W$ is chosen in such a way that $C=0$.

The multivalued function $e^{2\pi W}$ acquires the multiplier $e^{2\pi iT(B|[l])}$ after analytic continuation  along each closed loop $l$ in $M$. Therefore, $e^{2\pi W}$ is single-valued if and only if all periods $T(B|[l])$ ($[l]\in H_1(M,\mathbb{Z})$) of $B$ are integer. In the last case, $e^{2\pi W}|_\Gamma$ is an element of ${\rm Tr}\mathscr{A}(M)\equiv({\rm Ker}(H-i)\dotplus\mathbb{C})\cap C^{\infty}(\Gamma;\mathbb{C})$  due to Lemma \ref{traces of holom func lemma}. So, $B$ has integer periods only if the equation
\begin{equation}
\label{Main equation}
\partial_\gamma(H-i)\big[e^{-2\pi (H+i)f}\big]=0
\end{equation}
holds on $\Gamma$. Note that (\ref{Main equation}) is invariant under the replacement $f\mapsto f+q$, where $q$ is a smooth element of ${\rm Ker}(H+H^{-1})={\rm Ker}(H-i)\dotplus{\rm Ker}(H+i)\dotplus\mathbb{C}$. Indeed, if $q\in {\rm Ker}(H+i)\dotplus\mathbb{C}$, then $q$ is a boundary trace of some holomorphic function $w$ and $e^{-2\pi (H+i)q}=e^{-4\pi i q}$ is the trace of $e^{-4\pi i w}$. Then the condition $e^{-2\pi (H+i)f}\in {\rm Tr}\mathscr{A}(M)$ (equivalent to (\ref{Main equation})) implies $e^{-2\pi (H+i)(f+q)}=e^{-2\pi (H+i)f}e^{-4\pi i w}\in {\rm Tr}\mathscr{A}(M)$ and vice versa.

Now, suppose that (\ref{Main equation}) holds on $\Gamma$. Then there is a holomorphic function $w$ on $M$ whose boundary trace is equal to $e^{-2\pi (H+i)f}=e^{2\pi W}|_\Gamma$. Since $e^{2\pi W}$ and $w$ are holomorphic and $e^{2\pi W}=w$ on $\Gamma$, they coincide everywhere where one of them can be analytically continued. Thus, $e^{2\pi W}=w$ on $M$ and $e^{2\pi W}$ is single-valued. The latter means that $B$ has integer periods $T(B|[l])$ ($[l]\in H_1(M,\mathbb{Z})$). Thus, we arrive at the following statement.
\begin{prop}
\label{integer periods condition prop}
Introduce be the group
$$\mathscr{G}=\{B\in\mathcal{N} \ \ | \ T(B|[l])\in\mathbb{Z} \quad \forall [l]\in H_1(M,\mathbb{Z})\}$$
of vector fields with integer periods in $\mathcal{N}$ and denote by $\mathscr{G}_{\Gamma}=\mathfrak{T}(\mathscr{G})$ the corresponding group in $\mathcal{N}_\Gamma$. Then $\mathscr{G}_{\Gamma}$ can be determined from the DN map $\Lambda$ via the formula
$$\mathscr{G}_{\Gamma}=\{(B_\nu,f)\in\mathcal{N}_\Gamma \ \ | \ f \text{ is a solution to {\rm (}\ref{Main equation}{\rm)}}\}.$$
\end{prop}
Using representation (\ref{f via eigenfunctions}) for $f$, one can rewrite equation (\ref{Main equation}) in more convenient form. Let $\hat{f}_k=\overline{\hat{f}_{-k}}=\alpha_k+i\beta_k$, where $\alpha_k,\beta_k\in\mathbb{R}$ ($k=1,\dots,\mathfrak{g}$). Introduce the functions
\begin{equation}
\label{sepatrate factors}
\begin{split}
p_k:&={\rm exp}\big(-2\pi i[\eta_k(1+\mu_k)+\overline{\eta_k}(1-\mu_k)]\big), \\ 
q_k:&={\rm exp}\big(2\pi[\eta_k(\mu_k+1)+\overline{\eta_k}(\mu_k-1)]\big),
\end{split}
\end{equation}
where $\eta_k,\mu_k$ are given by (\ref{basis of eigenfunctions}) and (\ref{eigenvalues}). Then  
$$-2\pi i(H+i)f=2\pi\sum_k [ c_k\eta_k(\mu_k+1)-\overline{2\pi c_k\eta_k}(\mu_k-1)]$$ 
and (\ref{Main equation}) is equivalent to (\ref{Main equation on the coefficients}). As easily seen from (\ref{Main equation on the coefficients}), condition (\ref{Main equation}) is actually an equation on $2\mathfrak{g}$ real variables $\alpha_k,\beta_k$. Thus, $\varkappa:=(\alpha_1,\dots,\alpha_\mathfrak{g},\beta_1\dots,\beta_\mathfrak{g})$ is a solution to (\ref{Main equation}) if and only if
\begin{equation}
\label{parameters to data}
\begin{split}
(B_\nu(\varkappa),f(\varkappa))=\Big(i\sum_{k=1}^{\mathfrak{g}}[(\mu_k^{-1}-\mu_k)(\alpha_k+i\beta_k)\partial_\gamma\eta_k&+(\mu_k^{-1}+\mu_k)(\alpha_k-i\beta_k)\partial_\gamma\overline{\eta}_k],\\ 
&\sum_{k=1}^{\mathfrak{g}}[(\alpha_k+i\beta_k)\eta_k+(\alpha_k-i\beta_k)\overline{\eta}_k]\Big)
\end{split}
\end{equation}
is a boundary data of an element of $\mathscr{G}$.

\paragraph{Step 3. Determination of boundary data of canonical bases in $\mathcal{N}$.} As shown in Proposition  \ref{integer periods condition prop}, the solutions to (\ref{Main equation}) (or to (\ref{Main equation on the coefficients})) provide the boundary data of vectors with integer periods in $\mathcal{N}$. The next step is to find among them the boundary data $\mathfrak{T}(B_1),\dots,\mathfrak{T}(B_{2\mathfrak{g}})$ corresponding to some canonical basis $B_1,\dots,B_{2\mathfrak{g}}$ in $\mathcal{N}$. To this end, we apply the following statement. 
\begin{prop}
\label{improved criterion for canonicity prop}
\begin{enumerate}[label=\alph*{\rm)}]
\item Let $B_1,\dots,B_{2\mathfrak{g}}$ be a basis in $\mathcal{N}$ such that each field $B_k$ has integer periods. Then it is canonical {\rm(}i.e., dual to some canonical homology basis{\rm)} if and only if 
\begin{equation}
\label{canonicity criterion without duality}
(\Phi B_i,B_j)=(\Omega_{\mathfrak{g}})_{ij} \qquad \forall i,j=1,\dots,2\mathfrak{g}.
\end{equation}
\item Let $\kappa_1,\dots,\kappa_{2\mathfrak{g}}$ be elements of $\mathscr{G}_{\Gamma}$. Then $\mathfrak{T}^{-1}(\kappa_1),\dots,\mathfrak{T}^{-1}(\kappa_{2\mathfrak{g}})$ constitute canonical basis in $\mathcal{N}$ if and only if
\begin{equation}
\label{canonicity criterion on boundary}
{\bm\langle}\kappa_i,\kappa_j{\bm\rangle}=(\Omega_{\mathfrak{g}})_{ij} \qquad \forall i,j=1,\dots,2\mathfrak{g}
\end{equation}
{\rm(}the form ${\bm\langle}\cdot,\cdot{\bm\rangle}$ is given by {\rm(\ref{alterneting form}))}.
\end{enumerate}
\end{prop}
\begin{proof}
{\it a}) The necessity follows from  Lemma \ref{canonicity criterion lemma}, {\it b}). Let us prove the sufficiency. Let $Q_1,\dots,Q_{2\mathfrak{g}}$ be a canonical basis in $\mathcal{N}$ and let $[l_\cdot]$ be the corresponding canonical homology basis. Since $B_1,\dots,B_{2\mathfrak{g}}$ are linearly independent, we have $B_i=M_{ij}Q_j$, where $M$ is an invertible matrix. Then
$$T(B_i|[l_k])=\sum_{j}M_{ij}T(Q_j|[l_k])=\sum_{j}M_{ij}\delta_{jk}=M_{ik}$$
and, since each $B_i$ has integer periods, the entries of $M$ are integer. 

In view to Lemma \ref{inner product via periods lemma}, condition (\ref{canonicity criterion without duality}) implies
\begin{align*}
(\Omega_{\mathfrak{g}})_{ij}=(\Phi B_i,B_j)=-(B_i,\Phi B_j)=\sum_{ks}T(B_i|[l_k])(-\Omega_{\mathfrak{g}}^{-1})_{ks}T(-\Phi^2 B_j|[l_s])=\\
=\sum_{ks}T(B_i|[l_k])(\Omega_{\mathfrak{g}})_{ks}T(B_j|[l_s])=\sum_{ks}M_{ik}(\Omega_{\mathfrak{g}})_{ks}M^T_{sj}=(M\Omega_{\mathfrak{g}}M^T)_{ij}.
\end{align*}
Thus, we have $\Omega_{\mathfrak{g}}=M\Omega_{\mathfrak{g}} M^T$ and $0\ne {\rm det}(\Omega_{\mathfrak{g}})={\rm det}(\Omega_{\mathfrak{g}})({\rm det}(M))^2$, whence ${\rm det}(M)=\pm 1$ since the entries of $M$ are integer. Thus, entries of $M^{-1}$ are also integer and $B_1,\dots,B_{2\mathfrak{g}}$ is a basis in $\mathcal{N}$ dual to the homology basis 
$$[\tilde{l}_i]=(M^{-1})_{ij}[l_j].$$ Moreover, the new homology basis $[\tilde{l}_\cdot]$ is canonical due to Lemma \ref{canonicity criterion lemma}, {\it b}). 

{\it b}) Denote $B_k=\mathfrak{T}^{-1}(\kappa_k)$, then conditions (\ref{canonicity criterion without duality}) and (\ref{canonicity criterion on boundary}) are equivalent due to (\ref{alternating form isomorphism}). Therefore, {\it b}) follows from {\it a}).
\end{proof}

\paragraph{Step 4. Determination of period matrices of $M$ and $\mathbb{M}$.} Let $\kappa_1=\mathfrak{T}(B_1),\dots,\kappa_{2\mathfrak{g}}=\mathfrak{T}(B_{2\mathfrak{g}})$ be elements of $\mathcal{G}_\Gamma$ obeying condition (\ref{canonicity criterion on boundary}). In view of Proposition \ref{improved criterion for canonicity prop}, {\it b}), vectors $B_1,\dots,B_{2\mathfrak{g}}$ constitute a basis in $\mathcal{N}$ dual to some canonical homology basis $[l_\cdot]$ on $M$. In view of Lemma \ref{canonicity criterion lemma}, {\it c}) and Proposition \ref{determination of normal harmonic fields}, the auxiliary period matrix $\mathbb{B}$ corresponding to $[l_\cdot]$ obeys
\begin{equation}
\label{finding aux period matrix}
(\Omega_{(\mathfrak{g})}\mathfrak{B})_{ij}=(B_i,B_j)={\bm(}\kappa_i,\kappa_j{\bm)}
\end{equation}
(the inner product on $\mathcal{N}_\Gamma$ is defined by (\ref{inner product on boundary data})). 

So, using the previous steps and formula (\ref{finding aux period matrix}), one determines the auxiliary period matrix $\mathfrak{B}$ of $M$ corresponding to some canonical homology basis $[l_\cdot]=\{a_1,\dots,a_{\mathbb{g}},b_1,\dots,b_{\mathbb{g}}\}$ on it. Then the $b$-period matrix $\mathbb{B}$ of $\mathbb{M}$, corresponding to symmetric canonical basis (\ref{symmetric canonical basis}) is derived from $\mathfrak{B}$ by applying formulas (\ref{connection of period matrices}), (\ref{matrices}).

It remains to note that, although one cannot control the choice of $[l_\cdot]$, it is still possible to find all other auxiliary period matrices (corresponding to all possible canonical homology bases on $M$) by applying transformations (\ref{aux period matrix transformation rule}) to $\mathfrak{B}$. Then the substitution of these matrices into (\ref{connection of period matrices}), (\ref{matrices}) provides all symmetric $b$-period matrices of $\mathbb{M}$.

\subsection*{On the stability of the algorithm under small noise in the boundary data}
Let $\Lambda$ be a DN map of some (unknown) surface $(M,g)$ (we assume that the boundary $\Gamma$ of $(M,g)$ is given). We now explain the implementation of Steps 1-4 for the case in which only some approximation $\Lambda'$ of $\Lambda$ is known. Namely, we assume that $\Lambda'$ is a continuous operator acting from $H^1(\Gamma;\mathbb{C})$ to $L_2^{\mathbb{C}}(\Gamma;dl)$ and obeying 
\begin{equation}
\label{noise level}
\|\Lambda-\Lambda'\|_{H^1(\Gamma;\mathbb{C})\to L_2^{\mathbb{C}}(\Gamma;dl)}\le \varepsilon,
\end{equation}
where $\varepsilon$ is a small parameter called the noise bound. In what follows, we suppose that the noise bound is known to the one who applies Steps 1-4. 

Now, we describe the implementation of Steps 1-4 to obtain the approximation of some $b$-period matrix $\mathbb{B}$ of the double of $(M,g)$ via $\Lambda'$.

\paragraph{Step 1 (implementation).} Introduce the approximate Hilbert transform $H'=\Lambda^{'-1}\partial_\gamma$. In view of (\ref{noise level}), the operator $H^{'-1}=\partial_\gamma^{-1}\Lambda'$ obeys 
$$\|H^{'-1}-H^{-1}\|_{H^{1}(\Gamma;\mathbb{C})\mapsto H^{1}(\Gamma;\mathbb{C})}=O(\varepsilon).$$ 
Here and in the subsequent, all estimates are assumed to be uniform in $\Lambda'$ (but not uniform in $\Lambda$). In particular, the spectrum ${\rm Sp}(H^{'-1})$ of $H^{'-1}$ is contained in the $O(\varepsilon)$-neighborhood of the spectrum ${\rm Sp}(H^{-1})$ of $H^{-1}$. The essential spectrum of $H^{-1}$ is $\{i\}\cup\{-i\}$. Since the set of Fredholm operators is open in the operator norm (see Theorem 1.4.17, \cite{Murphy}), the essential spectrum of $H^{'-1}$ is contained in the $O(\varepsilon)$-neighborhoods of $\pm i$. The same estimates are valid for the spectra of $H'$,$H$. 

To find the approximations for eigenvalues (\ref{eigenvalues}) and eigenfunctions (\ref{basis of eigenfunctions}), we apply the following simple lemma.
\begin{lemma}
Suppose that $\lambda$ is a regular eigenvalue of a continuous operator $A$ {\rm(}acting in some Banach space $E${\rm)} of finite multiplicity and there is the punctured $c_0$-neighborhood of $\lambda$ which does not intersect the spectrum of $A$. Let $A'$ be an arbitrary continuous operator in $E$ obeying $\|A'-A\|<\varepsilon$ for sufficiently small $\varepsilon$. Let $(\lambda',f')$ be any eigenpair of $A'$ obeying $|\lambda'-\lambda|<\varepsilon$ and $\|f'\|=1$. Then there is $f\in {\rm Ker}(A-\lambda)$ such that $\|f-f'\|=O(\varepsilon)$. 
\end{lemma}
\begin{proof}
Consider the decomposition $E={\rm Ker}(A-\lambda)\dot{+}\tilde{E}$, where $\tilde{E}$ is a closed subspace in $E$ since ${\rm Ker}(A-\lambda)$ is finite-dimensional. Since $A$ is continuous, $(A-\lambda)\tilde{E}$ is closed and the operator $\tilde{A}=((A-\lambda)|_{\tilde{E}})^{-1}: \ (A-\lambda)\tilde{E}\to\tilde{E}$ is continuous due to the closed graph theorem. Decompose $f'$ as $f'=f+\tilde{f}$, where $f\in {\rm Ker}(A-\lambda)$ and $\tilde{f}\in\tilde{E}$. Since
$$\tilde{f}=\tilde{A}(A-\lambda)f'=\tilde{A}(A-A')f'-(\lambda-\lambda')\tilde{A}f',$$
we have $\|\tilde{f}\|=O(\varepsilon)$.
\end{proof}
Thus, to construct approximations of the eigenfunctions of $H$ corresponding to the unknown eigenvalue $\lambda=\lambda_i$, we find all (normalized in $L_2^{\mathbb{C}}(\Gamma;dl)$) eigenfunctions of $H'$ corresponding to the (nonzero) eigenvalues $\lambda'_k$ obeying $|\lambda'\pm i|>\sqrt{\varepsilon}$ and $|\lambda'_k-\lambda'_l|<\sqrt{\varepsilon}$ and then chose among them the maximal collection of pairwise orthogonal (in $L_2^{\mathbb{C}}(\Gamma;dl)$) eigenfunctions $\eta'_i,\dots,\eta'_{i+m(i)}$. As a result, for sufficiently small $\varepsilon$, we obtain the approximations $(\lambda'_{\pm k},\eta'_{\pm k})$ of the eigenpairs $(\lambda_{\pm k},\eta_{\pm k})$ obeying
\begin{equation}
\label{approx eigenpairs}
|\lambda'_{\pm k}-\lambda_{\pm k}|+\|\eta'_{\pm k}-\eta_{\pm k}\|_{H^{1}(\Gamma;\mathbb{C})}=O(\varepsilon) \qquad (k=1,\dots,\mathfrak{g}).
\end{equation}
Now, we introduce the space $\mathcal{N}'_\Gamma$ and the bilinear forms $(\cdot,\cdot)'$, $\langle\cdot,\cdot\rangle$ in the same way as in Proposition \ref{determination of normal harmonic fields}, where $\eta_{\pm k}$ are replaced by $\eta'_{\pm k}$. Denote $\kappa'_{\pm k}:=(-\partial_\gamma(H'+H^{'-1})\eta'_{\pm k},\eta'_{\pm k})$, then formulas (\ref{noise level}), (\ref{approx eigenpairs}) imply the closeness between the structures on $\mathcal{N}'_\Gamma$ and $\mathcal{N}_\Gamma$,
\begin{equation}
\label{closeness of the structures}
(\kappa'_{\pm i},\kappa'_{(\pm) j})'-(\kappa_{\pm i},\kappa_{(\pm) j})=O(\varepsilon), \quad \langle\kappa'_{\pm i},\kappa'_{(\pm) j}\rangle'-\langle\kappa_{\pm i},\kappa_{(\pm) j}\rangle=O(\varepsilon) \qquad (k=1,\dots,\mathfrak{g}).
\end{equation}

\paragraph{Step 2 (implementation).} Instead of (\ref{Main equation on the coefficients}), we consider the (approximate) equation
$$\partial_\gamma(H'-i)\big[(p'_1)^{\alpha'_1}\dots (p'_{\mathfrak{g}})^{\alpha'_\mathfrak{g}}(q'_1)^{\beta'_1}\dots (q'_{\mathfrak{g}})^{\beta'_\mathfrak{g}}\big]=0,$$
where $p'_k$, $q'_k$ are given by formula (\ref{sepatrate factors}) with $\eta_k$ replaced by $\eta'_k$. Introduce the functions
\begin{align*}
\mathscr{E}(\kappa)&:=\|\partial_\gamma(H-i)\big[p_1^{\alpha_1}\dots p_{\mathfrak{g}}^{\alpha_\mathfrak{g}}q_1^{\beta_1}\dots q_{\mathfrak{g}}^{\beta_\mathfrak{g}}\big]\|_{L^{\mathbb{C}}_2(\Gamma;dl)},\\
\mathscr{E}'(\kappa')&:=\|\partial_\gamma(H'-i)\big[(p'_1)^{\alpha'_1}\dots (p'_{\mathfrak{g}})^{\alpha'_\mathfrak{g}}(q'_1)^{\beta'_1}\dots (q'_{\mathfrak{g}})^{\beta'_\mathfrak{g}}\big]\|_{L^{\mathbb{C}}_2(\Gamma;dl)},
\end{align*}
where $\varkappa:=(\alpha_1,\dots,\alpha_\mathfrak{g},\beta_1\dots,\beta_\mathfrak{g})$ and $\varkappa':=(\alpha'_1,\dots,\alpha'_\mathfrak{g},\beta'_1\dots,\beta'_\mathfrak{g})$.

Recall that the global minima (i.e., zeroes) of $\mathscr{E}$ correspond to the boundary data of elements of $\mathcal{N}$ with integer periods via (\ref{parameters to data}). Let $\mathscr{B}$ be a sufficiently large closed ball in the parameter space $\mathbb{R}^{2\mathfrak{g}}$ of $\varkappa$ whose interior contains the zeroes of $\mathscr{E}$ corresponding to the elements of some canonical dual basis in $\mathcal{N}$. Then estimates (\ref{approx eigenpairs}), (\ref{noise level}), formula (\ref{sepatrate factors}), and the continuity of the embedding $H^{1}(\Gamma;\mathbb{C})\subset C(\Gamma;\mathbb{C})$ yield
\begin{equation}
\label{closenes of error functions}
\|\mathscr{E}'-\mathscr{E}\|_{C(\mathbb{R}^{2\mathfrak{g}})}=O(\varepsilon).
\end{equation}

Let $\varkappa\in\mathscr{B}$ be a zero of $\mathscr{E}$, then $\Pi=p_1^{\alpha_1}\dots p_{\mathfrak{g}}^{\alpha_\mathfrak{g}}q_1^{\beta_1}\dots q_{\mathfrak{g}}^{\beta_\mathfrak{g}}$ is a trace on $\Gamma$ of holomorphic invertible function on $(M,g)$. For small variations $\delta\varkappa:=(\delta\alpha_1,\dots,\delta\alpha_\mathfrak{g},\delta\beta_1,\dots,\delta\beta_\mathfrak{g})$, we have
\begin{align*}
\partial_\gamma(H-i)\big[p_1^{\alpha_1+\delta\alpha_1}\dots  &p_{\mathfrak{g}}^{\alpha_\mathfrak{g}+\delta\alpha_\mathfrak{g}}q_1^{\beta_1+\delta\beta_1}\dots q_{\mathfrak{g}}^{\beta_\mathfrak{g}+\delta\beta_\mathfrak{g}}\big]=\\
=&\partial_\gamma(H-i)\sum_k\big[{\rm log}p_k\delta \alpha_k+{\rm log}q_k\delta \beta_k\big]\Pi+O(|\delta\varkappa|^2).
\end{align*}
Note that the first term vanishes only if $\delta\varkappa=0$. Indeed, otherwise, there is the nonzero linear combination of $\Pi{\rm log}p_k$ and  $\Pi{\rm log}q_k$ which is a trace of holomorphic function. Since $\Pi^{-1}$ is a trace of holomorphic function and ${\rm log}p_k,{\rm log}q_k$ admit representations (\ref{sepatrate factors}), we conclude that there is a nonzero linear combination of $\eta_{\pm k}$ which is a trace of holomorphic function (i.e., an element of ${\rm Ker}(H-i)\dot{+}\mathbb{C}$), a contradiction. Thus, we obtain the non-degeneracy of all minima $\kappa$ of $\mathscr{E}$,
$$0<c_0<\frac{\mathscr{E}(\varkappa+\delta\varkappa)}{|\delta\varkappa|}<c_1<+\infty \qquad (|\delta\varkappa|<c_3),$$
where the constants $c_1,c_2,c_3$ depend on $\Lambda$ and $\mathscr{B}$. In particular, the inequality $|\mathscr{E}(\varkappa')|<\epsilon\ll 1$ implies that $\varkappa'$ lies in $O(\epsilon)$-neighborhood of some solution $\varkappa$ to (\ref{Main equation on the coefficients}).

Let us find the minimum $\varkappa'\in\mathscr{B}$ of $\mathscr{E}'$. Then $|\mathscr{E}'(\varkappa')|<\varepsilon$ and (\ref{closenes of error functions}) yields $|\mathscr{E}(\varkappa)|=O(\varepsilon)$. Thus, $|\varkappa'-\varkappa|=O(\varepsilon)$, where $\varkappa$ is a solution to (\ref{Main equation on the coefficients}). Now, remove from $\mathscr{B}$ the $\sqrt{\varepsilon}$-neighborhood of $\varkappa'$ and repeat the procedure, etc. As a result, we find all the approximations of the solutions to (\ref{Main equation on the coefficients}) in $\mathscr{B}$. For each approximation $\varkappa'$, we construct the approximate boundary data $\kappa':=(B'_\nu(\varkappa'),f'(\varkappa'))$ via formula (\ref{parameters to data}) with $\alpha_k,\beta_k,\eta_k,\mu_k$ replaced by $\alpha'_k,\beta'_k,\eta'_k,\mu'_k$, respectively. As a result, we obtain approximations $\kappa'$ of all boundary data $\kappa$ (with parameters $\varkappa$ in $\mathscr{B}$) obeying
\begin{equation}
\label{closeness of BD with IP}
\|\kappa'-\kappa\|_{L^\mathbb{C}_2(\Gamma;dl)\times H^{1}(\Gamma;\mathbb{C})}=O(\varepsilon).
\end{equation}
Since the radius of $\mathscr{B}$ is unknown, we actually start with some ball $\mathscr{B}'$, then enlarge it and repeat the above procedure, e.t.c., until we obtain the sufficiently large number of solutions $\kappa'$ to successfully perform the next step (finding the approximation of the boundary data of some canonical dual basis).

\paragraph{Step 3 (implementation).} Let us find the collection of the approximations $\kappa'_1,\dots,\kappa'_{2\mathfrak{g}}$ obeying
$$|{\bm\langle}\kappa'_i,\kappa'_j{\bm\rangle}'-(\Omega_{\mathfrak{g}})_{ij}|<\sqrt{\varepsilon}.$$
In view of (\ref{closeness of BD with IP}) and (\ref{closeness of the structures}), the corresponding exact boundary data $\kappa_1,\dots,\kappa_{2\mathfrak{g}}$ obey condition (\ref{canonicity criterion on boundary}) up to the discrepancy $O(\sqrt{\varepsilon})$. Since the left-hand side of (\ref{canonicity criterion on boundary}) is integer, this means that $\kappa_1,\dots,\kappa_{2\mathfrak{g}}$ constitute the boundary data of some canonical dual basis. 

\paragraph{Step 4 (implementation).} Let us calculate the $(2\mathfrak{g}\times 2\mathfrak{g})$-matrix $\mathfrak{P}'$ with the entries $\mathfrak{P}'_{ij}:=(\kappa'_i,\kappa'_j)'$. Then formula (\ref{finding aux period matrix}) and estimates (\ref{closeness of the structures}), (\ref{closeness of BD with IP}) imply that $\mathfrak{P}'-\Omega_{(\mathfrak{g})}\mathfrak{B}=O(\varepsilon)$, where $\mathfrak{B}$ is some auxilliary period matrix of $(M,g)$. Thus, we obtain the approximation $\mathfrak{B}'=\Omega_{(\mathfrak{g})}^{-1}\mathfrak{P}'$ of $\mathfrak{B}$ obeying $\mathfrak{B}'-\mathfrak{B}=O(\varepsilon)$. Now the substitution of $\mathfrak{B}'$ instead of $\mathfrak{B}$ into (\ref{connection of period matrices}) provides the approximation $\mathbb{B}'$ of some $b$-period matrix $\mathbb{B}$ of the double $\mathbb{M}$ of $(M,g)$, obeying $\mathbb{B}'-\mathbb{B}=O(\varepsilon)$. Thereby, Proposition \ref{convergence prop} is proved.

\subsection*{Statements and Declarations}
\paragraph{Competing Interests.} The author declares that there are no conflicts of interests and competing interests related to the present work.
\paragraph{Data Availibility Statement.} Data sharing not applicable to this article as no datasets were generated or analysed during the current study.

\end{document}